\documentclass[a4paper,american]{lipics-v2019}

\usepackage{amsopn}
\usepackage{amsmath,tabu}
\usepackage{amsfonts}
\usepackage{complexity}
\usepackage{mathtools}
\usepackage{tabularx}
\usepackage{environ}
\usepackage{hyperref}
\usepackage{tikz}
\usetikzlibrary{positioning}

\newcommand{\setcon}[1]{\lbrace #1 \rbrace}
\newcommand{\Setcon}[2]{\lbrace #1 \mid #2 \rbrace}
\newcommand{\Naturals}{\mathbb{N}}

\newcommand{\abs}[1]{|#1|}

\newcommand{\rel}{\mathit{rel}}
\newcommand{\po}{\mathit{po}}
\newcommand{\poloc}{\mathit{po} \, \text{-}\mathit{loc}}
\newcommand{\rf}{\mathit{rf}}

\newcommand{\tw}{\mathit{tw}}
\newcommand{\ww}{\mathit{ww}}
\newcommand{\cf}{\mathit{cf}}
\newcommand{\fr}{\mathit{fr}}
\newcommand{\dep}{\mathit{dp}}

\newcommand{\Lab}{\!\mathit{Lab}}
\newcommand{\Var}{\!\mathit{Var}}
\newcommand{\Val}{\!\mathit{Val}}
\newcommand{\EV}{E}
\newcommand{\WR}{\!\mathit{WR}}
\newcommand{\RD}{\!\mathit{RD}}
\newcommand{\rd}{\mathit{rd}}
\newcommand{\wri}{\mathit{wr}}
\newcommand{\var}{\mathit{var}}

\newcommand{\history}{\langle O, \po, \rf \rangle}
\newcommand{\graph}[1]{G_{\mathit{#1}}}
\newcommand{\MMcons}{\MM-\emph{Consistency}}

\newcommand{\MM}{\ComplexityFont{MM}}
\newcommand{\pomm}{\mathit{po} \, \text{-} \mathit{mm}}
\newcommand{\rfmm}{\mathit{rf} \text{-} \mathit{mm}}

\newcommand{\posc}{\mathit{po} \, \text{-} \mathit{sc}}
\newcommand{\rfsc}{\mathit{rf} \text{-} \mathit{sc}}

\newcommand{\TSO}{\ComplexityFont{TSO}}
\newcommand{\potso}{\mathit{po} \, \text{-} \mathit{tso}}
\newcommand{\rftso}{\mathit{rf} \text{-} \mathit{tso}}

\newcommand{\PSO}{\ComplexityFont{PSO}}
\newcommand{\popso}{\mathit{po} \, \text{-} \mathit{pso}}
\newcommand{\rfpso}{\mathit{rf} \text{-} \mathit{pso}}

\newcommand{\RMO}{\ComplexityFont{RMO}}
\newcommand{\pormo}{\mathit{po} \, \text{-} \mathit{rmo}}
\newcommand{\rfrmo}{\mathit{rf} \text{-} \mathit{rmo}}

\newcommand{\CC}{\ComplexityFont{CC}}
\newcommand{\PRAM}{\ComplexityFont{PRAM}}
\newcommand{\LOCAL}{\ComplexityFont{LOCAL}}
\newcommand{\POWER}{\ComplexityFont{POWER}}

\newcommand{\torder}{\mathit{t}}
\newcommand{\rorder}{\mathit{r}}
\newcommand{\twWR}{\tw[\WR]}
\newcommand{\Jgraph}[1]{J_{\mathit{#1}}}
\newcommand{\inc}{\mathit{inc}}
\newcommand{\cfinc}{\mathit{cfinc}}

\newcommand{\ext}{\mathit{ext}}
\newcommand{\ordmm}{\mathit{ord} \, \text{-} \mathit{mm}}
\newcommand{\llh}{\mathit{llh}}

\newcommand{\bigO}{\mathcal{O}}
\newcommand{\bigOS}{\bigO^*}

\newcommand{\ETH}{\ComplexityFont{ETH}}
\newcommand{\kSAT}[1]{#1\text{-}\SAT}

\newcommand{\Nxt} {\texttt{Nxt}}
\DeclareMathOperator{\First}{First}
\DeclareMathOperator{\Second}{Sec}
\DeclareMathOperator{\Lin}{Lin}

\makeatletter
\newcommand{\problemtitle}[1]{\gdef\@problemtitle{#1}}
\newcommand{\problemshort}[1]{\gdef\@problemshort{#1}}
\newcommand{\probleminput}[1]{\gdef\@probleminput{#1}}
\newcommand{\problemparameter}[1]{\gdef\@problemparameter{#1}}
\newcommand{\problemquestion}[1]{\gdef\@problemquestion{#1}}

\NewEnviron{myproblem}{
	\problemtitle{}
	\problemshort{}
	\probleminput{}
	\problemquestion{}
	
	\BODY
	\par\addvspace{.5\baselineskip}
	\noindent
	\framebox[\textwidth]{
		\begin{tabularx}{\textwidth}{@{\hspace{\parindent}} l X c}
			\multicolumn{2}{@{\hspace{\parindent}}l}{ \normalsize \emph{\@problemtitle} \@problemshort} \\
			\normalsize \textbf{Input:} & \normalsize \@probleminput \\
			\normalsize \textbf{Question:} & \normalsize \@problemquestion
		\end{tabularx}
	}
	\par\addvspace{.5\baselineskip}
}

\theoremstyle{definition}
\newtheorem{defi}[theorem]{Definition}

\title{A Framework for Consistency Algorithms}

\author{Peter Chini}{TU Braunschweig}{p.chini@tu-braunschweig.de}{}{}
\author{Prakash Saivasan}{The Institute of Mathematical Sciences}{prakashs@imsc.res.in}{}{}
 
\authorrunning{P. Chini and P. Saivasan}

\Copyright{Peter Chini and Prakash Saivasan}



\nolinenumbers

\begin{document}

\maketitle

\begin{abstract}
	We present a framework that provides deterministic consistency algorithms for given memory models.
	Such an algorithm checks whether the executions of a shared-memory concurrent program are consistent under the axioms defined by a model.
	For memory models like $\SC$ and $\TSO$, checking consistency is $\NP$-complete.
	Our framework shows, that despite the hardness, fast deterministic consistency algorithms can be obtained by employing tools from fine-grained complexity.
	
	The framework is based on a universal consistency problem which can be instantiated by different memory models.
	We construct an algorithm for the problem running in time $\bigOS(2^k)$, where $k$ is the number of write accesses in the execution that is checked for consistency.
	Each instance of the framework then admits an $\bigOS(2^k)$-time consistency algorithm.
	By applying the framework, we obtain corresponding consistency algorithms for $\SC$, $\TSO$, $\PSO$, and $\RMO$.
	Moreover, we show that the obtained algorithms for $\SC$, $\TSO$, and $\PSO$ are optimal in the fine-grained sense:
	there is no consistency algorithm for these running in time $2^{o(k)}$ unless the exponential time hypothesis fails.
\end{abstract}

\section{Introduction}
\label{Section:Introduction}

The paper at hand develops a framework for consistency algorithms.
Given an execution of a concurrent program over a shared-memory system, consistency algorithms check whether the execution is consistent under the intended behavior of the memory.
Our framework takes an abstraction of this intended behavior, a \emph{memory model}, and yields a deterministic consistency algorithm for it.
By applying the framework, we obtain provably optimal consistency algorithms for the well-known memory models $\SC$ \cite{Lamport1979}, $\TSO$, and $\PSO$ \cite{Sparc}.

Checking consistency is central in the verification of shared-memory implementations.
Such implementations promise programmers consistency guarantees according to a certain memory model.
However, due to the complex and performance-oriented design, implementing shared memories is sensitive to errors and implementations may not provide the promised guarantees.
Consistency algorithms test this.
They take an execution over a shared-memory implementation, multiple sequences of read and write events, one for each thread. 
Then they check whether the execution is viable under the memory model, namely whether read and write events can be arranged in an interleaving that satisfies the axioms of the model.

In 1997, Gibbons and Korach \cite{Gibbons1997} were the first ones that studied consistency checking as it is considered in this work.
They focused on the basic memory model \emph{Sequential Consistency} ($\SC$) by Lamport \cite{Lamport1979}.
In $\SC$, read and write accesses to the memory are atomic making each write of a thread immediately visible to all other threads.
Gibbons and Korach showed that checking consistency in this setting is, in general, $\NP$-complete.
Moreover, they considered restrictions of the problem showing that even under the assumption that certain parameters like the number of threads are constant, the problem still remains $\NP$-complete.

The SPARC memory models \emph{Total Store Order} ($\TSO$), \emph{Partial Store Order} ($\PSO$), and \emph{Relaxed Memory Order} ($\RMO$) were investigated by Cantin et al. in \cite{Cantin2005}.
The authors showed that, like for $\SC$, checking consistency for these models is $\NP$-hard.
Furbach et al. \cite{Furbach2015} extended the $\NP$-hardness to almost all models appearing in the Steinke-Nutt hierarchy \cite{Steinke2004}, a hierarchy developed for the classification of memory models.
This yields $\NP$-hardness results for memory models like \emph{Causal Consistency} ($\CC$) \cite{Lamport1978}, \emph{Pipelined RAM} ($\PRAM$) \cite{Lipton1988}, \emph{Cache Consistency} \cite{Goodman1991} or variants of \emph{Processor Consistency} \cite{Goodman1991,Ahamad1993}.
Bouajjani et al. \cite{Bouajjani2017} independently found that checking $\CC$, and variants of it, for a given execution is $\NP$-hard.

We approach consistency checking under the assumption of data-independence \cite{Bouajjani2017,Bouajjani2019,Biswas2019}.
In fact, the behavior of a shared-memory implementation or a database does not depend on precise values in practice \cite{Wolper1986,Jepsen,Abdulla2013}.
We can therefore assume that in a given execution, a value is written at most once.
However, the $\NP$-hardness of checking consistency under $\SC$, $\TSO$, and $\PSO$ carries over to the data-independent case \cite{Gibbons1997,Furbach2015}.
Deterministic consistency algorithms for these models will therefore face exponential running times.
By employing a \emph{fine-grained} complexity analysis, we show that one can still obtain consistency algorithms that have only a \emph{mild} exponential dependence on certain parameters and are provably \emph{optimal}.

Fine-grained complexity analyses are a task of \emph{Parameterized Complexity} \cite{Fomin2010,Cygan2015,Downey2013}.
The goal of this new field within complexity theory is to measure the influence of certain parameters on a problem's complexity.
In particular, if a problem is $\NP$-hard, one can determine which parameter $k$ of the problem still offers a fast deterministic algorithm.
Such an algorithm runs in time $f(k) \cdot \mathit{poly}(n)$, where $f$ is a computable function that only depends on the parameter, and $\mathit{poly}(n)$ is a polynomial dependent on the size of the input $n$.
Problems admitting such algorithms lie in the class $\FPT$ of \emph{fixed-parameter tractable} problems.
The time-complexity of a problem in $\FPT$ is denoted by $\bigOS(f(k))$.
A fine-grained complexity analysis determines the precise function $f$ that is needed to solve the problem.
While finding upper bounds amounts to finding algorithms, lower bounds on $f$ can be obtained from the \emph{exponential time hypothesis} ($\ETH$) \cite{Impagliazzo2001}.
It assumes that $n$-variable $\kSAT{3}$ cannot be solved in time $2^{o(n)}$ and is considered standard in parameterized complexity \cite{Cygan2015,Lokshtanov2011,Cygan2016,Chen2005}.
A function $f$ is \emph{optimal} when upper and lower bound match.

Our contribution is a framework which yields consistency algorithms that are optimal in the fine-grained sense.
Obtained algorithms run in time $\bigOS(2^k)$, where $k$ is the number of write events in the given execution.
We demonstrate the applicability by obtaining consistency algorithms for $\SC$, $\TSO$, $\PSO$, and $\RMO$.
Relying on the $\ETH$, we prove that for the former three models, consistency cannot be checked in time $2^{o(k)}$.
This shows that our framework yields optimal algorithms for these models.
Note that considering other parameters like the number of threads, the number of events per thread, or the size of the underlying data domain yields $\W[1]$-hard problems \cite{Pavlogiannis2020,Gibbons1997} that are unlikely to admit $\FPT$-algorithms \cite{Cygan2015,Downey2013}.

The framework is based on a universal consistency problem that can be instantiated by a memory model of choice.
We develop an algorithm for this universal problem running in time $\bigOS(2^k)$.
Then, any instance by a memory model automatically admits an $\bigOS(2^k)$-time consistency algorithm.
For the formulation of the problem, we rely on the formal framework of Alglave \cite{Alglave2012} and Alglave et al. \cite{Alglave2014} for describing memory models in terms of relations.
In fact, checking consistency then amounts to finding a particular \emph{store order} \cite{Bouajjani2019} on the write events that satisfies various acyclicity constraints.

For solving the universal consistency problem, we show that instead of a store order we can also find a total order on the write events satisfying similar acyclicity constraints.
The latter are algorithmically simpler to find.
We develop a notion of \emph{snapshot orders} that mimic total orders on subsets of write events.
This allows for shifting from the relation-based domain of the problem to the subset lattice of writes.
On this lattice, we can perform a dynamic programming which builds up total orders step by step and avoids an explicit iteration over such.
Keeping track of the acyclicity constraints is achieved by so-called \emph{coherence graphs}.
The dynamic programming runs in time $\bigOS(2^k)$ which constitutes the complexity.

To apply the framework, we follow the formal description of $\SC$, $\TSO$, $\PSO$, and $\RMO$, given in \cite{Alglave2012,Alglave2014} and instantiate the universal consistency problem.
Optimality of the algorithms for $\SC$, $\TSO$, and $\PSO$ is obtained from the $\ETH$.
To this end, we construct a reduction from $\kSAT{3}$ to the corresponding consistency problem that generates only linearly many write events.
The reduction transports the assumed lower bound on $\kSAT{3}$ to consistency checking.

\subparagraph*{Related Work.}
\label{Section:RelatedWork}

In its general form, consistency checking is $\NP$-hard for most memory models.
Furbach et al. \cite{Furbach2015} show that $\LOCAL$ \cite{Heddaya1992} is an exception.
Checking consistency under $\LOCAL$ takes polynomial time.
This also holds for \emph{Cache Consistency} and $\PRAM$ if certain parameters of the consistency problem are assumed to be constant.
In the case of data-independence, Bouajjani et al. \cite{Bouajjani2017} show that checking consistency under $\CC$ and variants of $\CC$ also takes polynomial time.
Wei et al. \cite{Wei2013} present a similar result for $\PRAM$.
In \cite{Bouajjani2019}, Bouajjani et al. present practically efficient algorithms for the consistency problems of $\SC$ and $\TSO$ under data-independence.
They rely on the polynomial-time algorithm for $\CC$ \cite{Bouajjani2017} and obtain a partial store order, which is completed by an enumeration.
In theory, the enumeration has a worst-case time complexity of $\bigOS(k^k)$.
We avoid such an enumeration by a dynamic programming running in time $\bigOS(2^k)$.
Consistency checking for weaker and stronger notions of consistency, like \emph{linearizability} \cite{Herlihy1990}, is considered in \cite{Enea2018,Emmi2015,Emmi2018}.

Instead of checking consistency for a single execution of a shared-memory implementation, there were efforts in verifying that all executions are consistent under a certain memory model.
Alur et al. show in \cite{Alur2000} that for $\SC$, the problem is undecidable.
This also holds for $\CC$ \cite{Bouajjani2017}.
Under data-independence, the problem becomes decidable for $\CC$ \cite{Bouajjani2017}.
Verifying \emph{Eventual Consistency} \cite{Terry1995} was shown to be decidable by Bouajjani et al. in \cite{Bouajjani2014}.
There has also been work on other verification problems like reachability and robustness.
Atig et al. show in \cite{Atig2010} that, under $\TSO$ and $\PSO$, reachability is decidable.
In \cite{Atig2012} the authors extend their results and present a relaxation of $\TSO$ with decidable reachability problem.
Robustness against $\TSO$ was considered in \cite{Bouajjani2011} and shown to be $\PSPACE$-complete.
This also holds for $\POWER$ \cite{Mador-Haim2012,Sarkar2011}, as shown in \cite{Meyer2014}, and for partitioned global address spaces~\cite{Calin2013}.

Parameterized complexity has been applied to other verification problems as well.
Biswas and Enea \cite{Biswas2019} study the complexity of transactional consistency and obtain an $\FPT$-algorithm in the size and the width of a history.
This also yields an algorithm for the serializability problem, proven to be $\NP$-hard by Papadimitriou \cite{Papadimitriou1979} in 1979.
A fine-grained algorithm for serializability under $\TSO$ was given in \cite{Farzan2016}.
The authors of \cite{Farzan2009} present an $\FPT$-algorithm for predicting atomicity violations as well as an intractability result.
The parameterized complexity of data race prediction was considered in \cite{Pavlogiannis2020}.
Fine-grained complexity analyses were conducted for reachability under \emph{bounded context switching} on finite-state systems \cite{Chini2017}, and for reachability and liveness on parameterized systems \cite{Chini2018,Chini2019}.

\section{Preliminaries}
\label{Section:Preliminaries}

To state our framework, we introduce some basic notions around memory models and the consistency problem.
We mainly follow \cite{Alglave2014,Alglave2012,Bouajjani2019,Bouajjani2017}.
Further, we give a short introduction into fine-grained complexity.
For standard textbooks in this field, we refer to  \cite{Fomin2010,Downey2013,Cygan2015}.

\subparagraph{Relations, Histories, and Memory Models.}
We consider the consistency problem: given an execution of a concurrent program and a model of the shared memory, decide whether the execution adheres to the model.
Formally, executions consist of \emph{events} modeling write and read accesses to the shared memory.
To define these, let $\Var$ be the finite set of variables of the program.
Moreover, let $\Val$ be its finite data domain and $\Lab$ a finite set of labels.
A \emph{write event} is defined by $w\!:\! \wri(x,v)$, where $w \in \Lab$ is a label, $x \in \Var$ is a variable, and $v \in \Val$ is a value.
The set of write events is defined by $\WR = \Setcon{w \!:\! \wri(x,v)}{w \in \Lab, x \in \Var, v \in \Val}$.
A \emph{read event} is given by $r \!:\! \rd(x,v)$.
The set of read events is denoted by $\RD$.
We define the set of all \emph{events} by $\EV = \WR \cup \RD$.
If it is clear from the context, we omit the label of an event.
Given an event $o \in \EV$, we access the variable of $o$ by $\var(o) \in \Var$.
For a subset $O \subseteq \EV$, we denote by $\WR(O)$ and $\RD(O)$ the set of write and read events in $O$.

For modeling dependencies between events we use strict orders.
Let $O \subseteq \EV$ be a set of events.
A \emph{strict partial order} on $O$ is an irreflexive, transitive relation over $O$.
A \emph{strict total order} is a strict partial order that is total.
We often refer to the notions without mentioning that they are strict.
Given two relations $\rel, \rel' \subseteq O \times O$, we denote by $\rel \circ \rel'$ their composition, by $\rel^+$ the transitive~closure, and by $\rel^{-1}$ the inverse.
For variable $x$, we denote by $\rel_x$ the restriction of $\rel$ to events on $x$: $\rel_x = \Setcon{(o,o') \in \rel}{\var(o) = \var(o') = x}$.

Executions are modeled by \emph{histories}.
A \emph{history} is a tuple $h = \history$, where $O \subseteq \EV$ is a set of events executed by the threads of the program.
The \emph{program order} $\po$ is a partial order on $O$ which orders the events of a thread according to the execution.
Typically, it is a union of total orders, one for each thread.
The relation $\rf \subseteq \WR(O) \! \times \RD(O)$ is called \emph{reads-from} relation.
It specifies the write event providing the value for a read event in the history.
Moreover, for each read event $r \in \RD(O)$ we have a write event $w \in \WR(O)$ such that $(w,r) \in \rf$ and if $(w,r) \in \rf$, both events access the same variable.

Note that we assume the reads-from relation to be given as a part of the history.
This is due to the data-independence of shared-memory and database implementations in practice \cite{Wolper1986,Jepsen,Biswas2019,Abdulla2013,Bouajjani2017,Bouajjani2019}.
This means that the behavior of the implementation does not depend on actual values and in an execution, we may assume each value to be written at most once.
From such an execution, we can simply read off the relation $\rf$.

Our framework is compatible with histories that feature \emph{initial writes}.
These histories have a write event for each variable writing the initial value of that variable.
Formally, these write events are smaller than all other events under program order.
If a history $h = \history$ is fixed, we abuse notation and also use $\WR$ and $\RD$ to denote $\WR(O)$ and $\RD(O)$.
For a variable $x$, we write $\WR(x) = \Setcon{w \in \WR}{\var(w) = x}$ for the set of write events on $x$ in $h$.
Furthermore, we will later make use of the relation $\poloc$, defined by restricting $\po$ to events on the same variable: $\poloc = \Setcon{(o,o') \in \po}{\var(o) = \var(o')}$.

A \emph{memory model} is an abstraction of the memory behavior defining axioms that the relations in a history must adhere to.
Formally, a \emph{memory model} $\MM$ is a tuple \mbox{$\MM = (\pomm,\rfmm)$.}
The relation $\pomm$, also called \emph{preserved program order}, is a subrelation of $\po$ describing the structure maintained by the memory model.
The latter relation $\rfmm$ is a subrelation of $\rf$.
It shows which write events are visible globally under $\MM$.

\subparagraph*{Fine-Grained Complexity.}
For many memory models, the consistency problem is $\NP$-hard \cite{Furbach2015,Gibbons1997,Cantin2005,Bouajjani2017}.
Hence, deterministic consistency algorithms usually face exponential running times.
But exponents might only depend on certain parameters of the problem which still allow the algorithm for being fast.
Finding such parameters is a task of \emph{parameterized~complexity}.

The basis of parameterized complexity are \emph{parameterized problems}.
That is, subsets $P$ of $\Sigma^* \times \Naturals$, where $\Sigma$ is a finite alphabet.
An input to $P$ is of the form $(x,k)$, with $k$ being called the \emph{parameter}.
A particularly interesting class of parameterized problems are the \emph{fixed-parameter tractable} $(\FPT)$ problems.
A problem $P$ is $\FPT$ if it can be solved by a deterministic algorithm running in time $f(k) \cdot \abs{x}^{\bigO(1)}$, where $f$ is a computable function only dependent on $k$.
The running time of such an algorithm is usually denoted by $\bigOS(f(k))$ to suppress the polynomial part.
The class $\FPT$ is contained in the class $\W[1]$.
Problems that are $\W[1]$-hard are considered intractable since they are unlikely to be $\FPT$.

Given a fixed-parameter tractable problem $P$, finding an upper bound for $f$ is achieved by constructing an algorithm for $P$.
Lower bounds on $f$ are usually obtained from the \emph{exponential time hypothesis} ($\ETH$) \cite{Impagliazzo2001}.
This standard hardness assumptions asserts that $\kSAT{3}$ cannot be solved by an algorithm running in time $2^{o(n)}$, where $n$ is the number of variables.
A lower bound on $f$ is then obtained by a suitable reduction from $\kSAT{3}$ to $P$.
We are interested in finding the \emph{optimal} $f$ for the consistency problem where upper and lower bound match.
The search for such an $f$ is referred to as \emph{fine-grained complexity}.

\section{Framework}
\label{Section:Framework}

We present our framework.
Given a model describing the memory, the framework provides an (optimal) deterministic algorithm for the corresponding consistency problem.
That is, whether a given history can be scheduled under the axioms imposed by the model.
The obtained algorithm can then be used within a testing routine for concurrent programs.

At the heart of the framework is a consistency problem that can be instantiated with different memory models.
We solve this universal problem by switching from a relation-based domain, where the problem is defined, to a subset-based domain.
On the latter, we can then apply a dynamic programming which constitutes the desired deterministic algorithm.

\subsection{Universal Consistency}
\label{Section:Universal}

The basis of our framework is a universal consistency problem which can be instantiated to simulate a particular memory model.
For its formulation, we make use of a consistency notion that allows for the construction of a fast algorithm but deviates from the literature~\cite{Alglave2012,Alglave2014,Bouajjani2019} at first sight.
Therefore, it is proven in Section \ref{Section:Instantiations} that instantiating the problem with a particular memory model yields the correct notion of consistency.

We clarify our notion of consistency.
Intuitively, a history is consistent under a memory model if it can be scheduled such that certain axioms defined by the model are satisfied.
Following the formal framework of \cite{Alglave2012,Alglave2014}, finding such a schedule amounts to finding a particular order of the write events that satisfies acyclicity requirements imposed by the axioms.
Formally, let $h = \history$ be a history and let $\MM$ be a memory model described by the tuple $(\pomm,\rfmm)$.
Then $h$ is called $\MM$\emph{-consistent} if there exists a strict total order $\tw$ on the write events $\WR$ of $h$ such that the graphs
\begin{align*}
	\graph{loc} = (O, \poloc \cup \rf \cup \tw \cup \cf)
	\ \ \text{and} \ \ 
	\graph{mm} = (O, \pomm \cup \rfmm \cup \tw \cup \cf)
\end{align*}
are both acyclic.
Here, the \emph{conflict relation} $\cf$ is defined by $\cf = \rf^{-1} \circ \bigcup_{x \in \Var} \tw_x$.
Phrased differently, $(r,w) \in \cf$ if $r$ is a read event on a variable $x$, $w$ is a write event on $x$, and there is a write event $w'$ on $x$ such that $(w',r) \in \rf$ and $(w',w) \in \tw$.

The acyclicity of $\graph{loc}$ is called \emph{uniprocessor} requirement \cite{Alglave2012} or \emph{memory coherence} for each location \cite{Cantin2005}.
Roughly, it demands that an order among writes to the same location that can be extracted from the history, is kept in $\tw$.
The second acyclicity requirement in the definition resembles the underlying memory model $\MM$.
If $\graph{mm}$ is acyclic, the history can be scheduled adhering to the axioms defined by $\MM$.

Our definition of consistency deviates from the literature in two aspects.
First, we demand a total order $\tw$ instead of a \emph{store order}, a partial order that is total on writes to the same location \cite{Alglave2012,Alglave2014,Bouajjani2019}.
In Section \ref{Section:Instantiations} we will show that the resulting notions of consistency are equivalent.
A further difference is that we do not explicitly test for \emph{out of thin air values}~\cite{Manson2005}.
For the majority of memory models considered in this work, the test is not necessary as it is implied by the acyclicity of $\graph{loc}$ and $\graph{mm}$.
But it can easily be added when needed.

We are ready to state the universal consistency problem.
To this end, let $\MM$ be a fixed memory model.
Given a history $h$, the problem asks whether $h$ is $\MM$-consistent.
\begin{myproblem}
	\problemtitle{$\MM$-Consistency}
	\probleminput{A history $h = \langle O, \po, \rf \rangle$.}
	\problemquestion{Is $h$ $\MM$-consistent?}
\end{myproblem}
Instantiations of the problem by well-known memory models like $\SC$ or $\TSO$ are typically $\NP$-hard~\cite{Gibbons1997,Furbach2015}.
However, we are interested in a deterministic algorithm for \MMcons.
While we cannot avoid an exponential running time for such an algorithm, a fine-grained complexity analysis can determine the \emph{optimal} exponential dependence.
Many parameters of \MMcons~like the number of threads, the maximum size per thread, or the size of the data domain yield parameterizations that are $\W[1]$-hard \cite{Pavlogiannis2020,Gibbons1997}.
Therefore, we conduct a fine-grained analysis for the parameter $k = \abs{\WR}$, the number of writes in $h$.
The main finding is an algorithm for \MMcons~running in time $\bigOS(2^k)$.
The optimality of this approach is shown in Section \ref{Section:LowerBounds} by a complementing lower bound.
We formally state the upper bound in the following theorem.
There, $n = \abs{O}$ denotes the number of events in $h$.
\begin{theorem}
	\label{Theorem:MMConsistency}
	The problem $\MM$-Consistency can be solved in time $\bigO(2^k \cdot k^2 \cdot n^2)$.
\end{theorem}

Note that an algorithm for \MMcons~running in time $\bigOS(k^k)$ is immediate.
One can iterate over all total orders of $\WR$ and check the acyclicity of $\graph{loc}$ and $\graph{mm}$ in polynomial time.
Since we cannot afford this iteration in $\bigOS(2^k)$, improving the running time needs an alternative approach and further technical development that we summarize in Section \ref{Section:Algorithm}.

\subsection{Algorithm}
\label{Section:Algorithm}

We present the upper bound for \MMcons~as stated in Theorem \ref{Theorem:MMConsistency}.
Our algorithm is a dynamic programming. 
It switches from the domain of total orders to subsets of write events and iterates over the latter.
The crux is that for a particular subset we do not need to remember a precise order.
In fact, we only need to store that it can be ordered by a so-called \emph{snapshot order} that mimics total orders on subsets.
Not having a precise order at hand yields a disadvantage: we cannot just test both acyclicity requirements in the end.
Instead, we perform an acyclicity test on a \emph{coherence graph} in each step of the iteration.
These graphs carry enough information to ensure acyclicity as it is required by \MMcons.

We begin our technical development by introducing \emph{snapshot orders}.
Intuitively, these simulate total orders of the write events on subsets of writes.
Given a subset, a \emph{snapshot order} consists of two parts: a total order on the subset and a partial order.
The latter expresses that the complement of the given set precedes the subset but is yet unordered.
\begin{defi}
	Let $V \subseteq \WR$.
	A \emph{snapshot order} on $V$ is a union $\tw[V] = \torder[V] \cup \rorder[V]$.
\end{defi}
The relation $\torder[V]$ is a strict total order on $V$ and $\rorder[V] = \Setcon{(\overline{v},v)}{\overline{v} \in \overline{V}, v \in V}$ arranges that the elements of $\overline{V}$ are smaller than the elements of $V$.
By $\overline{V}$, we denote the complement of $V$ in the write events, $\overline{V} = \WR \setminus V$.
Note that $\rorder[V]$ does not impose an order among $\overline{V}$.

A snapshot order is indeed a strict partial order.
Even more, when the considered set is the whole write events $\WR$, a snapshot order $\twWR$ is a total order on $\WR$.
Therefore, $\MM$-consistency can be checked by finding a snapshot order on $\WR$ satisfying both acyclicity requirements.
The advantage of this formulation is that we can construct such an order from snapshot orders on subsets.
Technically, we parameterize\footnote{The parameterization here does not refer to parameterized complexity.} the problem along all $V \subseteq \WR$.

For the acyclicity requirements, we need a similar parameterization.
To this end, let $V \subseteq \WR$ be a subset and $\tw[V]$ a snapshot order on $V$.
We parameterize the above graphs $\graph{loc}$ and $\graph{mm}$ via exchanging the total order by the snapshot order:
\begin{align*}
	\graph{loc}(\tw[V]) &= (O, \poloc \cup \rf \cup \tw[V] \cup \cf[V]), \\ \graph{mm}(\tw[V]) &= (O, \pomm \cup \rfmm \cup \tw[V] \cup \cf[V]).
\end{align*}
As above, the conflict relation is defined by $\cf[V] = \rf^{-1} \circ \bigcup_{x \in \Var} \tw[V]_x$.
Note that for a snapshot order $\twWR$ on the whole set of write events, the resulting graphs $\graph{loc}(\twWR)$ and $\graph{mm}(\twWR)$ are exactly those appearing in the acyclicity requirement.

Now we have the tools to state the parameterization of \MMcons~along subsets of write events.
This allows for leaving the domain of total orders and switch to subsets instead.
To this end, we define a table $T$ with a Boolean entry $T[V]$ for each $V \subseteq \WR$.
Entry $T[V]$ will be $1$, if there is a snapshot order on $V$ satisfying the acyclicity requirement on both parameterized graphs.
Otherwise, $T[V]$ will evaluate to $0$.
Formally, $T[V]$ is defined by
\begin{align*}
	T[V] = \left\lbrace
	\begin{aligned}
		1,& ~\text{if}~ \exists ~\text{snapshot ord.}~ \tw[V] : \graph{loc}(\tw[V]) ~\text{and}~ \graph{mm}(\tw[V]) ~\text{are acyclic},\\
		0,& ~\text{otherwise}.
	\end{aligned}
	\right.
\end{align*}

The following lemma relates \MMcons~to the table $T$.
It is crucial in our development as it states the correctness of the constructed parameterization.
The proof follows from the beforehand definitions and the fact that a snapshot order on $\WR$ is already total.
\begin{lemma}
	\label{Lemma:MMParameterization}
	History $h$ is $\MM$-consistent if and only if $T[\WR] = 1$.
\end{lemma}

We are now left with the problem of evaluating the entry $T[\WR]$.
Our approach is to set up a recursion among the entries of $T$ and evaluate it via a bottom-up dynamic programming.
The recursion will explain how entries of subsets are aggregated to compute entries of larger sets.
In fact, write events are added element by element:
the recursion shows how an entry $T[V]$ can be utilized to compute the entry of an enlarged set $V \cup \setcon{v}$, where $v \in \overline{V}$.

When passing from $T[V]$ to $T[V \cup \setcon{v}]$, we need to provide a snapshot order on $V \cup \setcon{v}$ that satisfies the acyclicity requirements.
A snapshot order on $V$ can always be extended to a snapshot order on $V \cup \setcon{v}$:
we insert $v$ as new minimal element in the contained total order.
But we need to keep track of whether the acyclicity is compatible with the new minimal element $v$.
To this end, we perform acyclicity tests on \emph{coherence graphs}.
These do not depend on a snapshot order and solely rely on the fact that $v$ is the new minimal element.
This will later allow for an evaluation of the table without touching precise orders.
\begin{defi}
	Let $V \subseteq \WR$ and $v \in \overline{V}$.
	The \emph{coherence graphs} of $V$ and $v$ are defined by
	\begin{align*}
		\graph{loc}[V,v] &= (O, \poloc \cup \rf \cup \rorder[V,v] \cup \cf[V,v]), \\
		\graph{mm}[V,v] &= (O, \pomm \cup \rfmm \cup \rorder[V,v] \cup \cf[V,v]).
	\end{align*}
\end{defi}
In the definition, relation $\rorder[V,v]$ expresses that $\overline{V \cup \setcon{v}}$ is smaller than $V \cup \setcon{v}$ and that $v$ is the minimal element in $V \cup \setcon{v}$.
Formally, it is given by $\rorder[V,v] = \rorder[V \cup \setcon{v}] \cup \Setcon{(v,w)}{w \in V}$.
The conflict relation is defined by $\cf[V,v] = \rf^{-1} \circ \bigcup_{x \in \Var} \rorder[V,v]_x$.

Coherence graphs are key for the recursion among the entries of $T$.
Assume we are given a snapshot order $\tw[V]$ on $V$ meeting the acyclicity requirements of $T$ and we extend it to a snapshot order $\tw[V']$ on $V' = V \cup \setcon{v}$, as above - by inserting $v$ as minimal element of $V'$.
We show that each potential cycle in $\graph{loc}(\tw[V'])$ or $\graph{mm}(\tw[V'])$ either implies a cycle in a coherence graph $\graph{loc}[V,v]$ or $\graph{mm}[V,v]$ or in one of the graphs $\graph{loc}(\tw[V])$ or $\graph{mm}(\tw[V])$.
If $T[V] = 1$, we can assume the latter graphs to be acyclic.
Moreover, if we have checked that the coherence graphs are acyclic as well, we obtain that $T[V'] = 1$.
Hence, a recursion should check whether $T[V] = 1$ and whether the corresponding coherence graphs are acyclic.

We formulate the recursion in the subsequent lemma.
Note that it is a top-down formulation that only refers to non-empty subsets of write events. 
An evaluation of the base case is immediate. 
Entry $T[\emptyset]$ is evaluated to $1$ if $\graph{loc}(\emptyset) = (O, \poloc \cup \rf)$ and $\graph{mm}(\emptyset) = (O,\pomm \cup \rfmm)$ are both acyclic.
Otherwise it is evaluated to $0$.
\begin{lemma}
	\label{Lemma:Recursion}
	Let $V \subseteq \WR$ be a non-empty subset.
	Entry $T[V]$ admits the following recursion:
	\begin{align*}
		T[V] = \bigvee_{v \in V} 
		\left(\graph{loc}[V \! \setminus \! \setcon{v},v] ~\text{acyclic}\right) \wedge
		\left(\graph{mm}[V \! \setminus \! \setcon{v},v] ~\text{acyclic}\right)
		\wedge
		T[V \! \setminus \! \setcon{v}].
	\end{align*}
\end{lemma}
We interpret $\left(\graph{loc}[V \! \setminus \! \setcon{v},v] ~\text{\emph{acyclic}}\right)$ as a predicate evaluating to $1$ if the graph is acyclic, to $0$ otherwise.
Hence, the recursion requires the existence of an $v \in V$ such that both coherence graphs are acyclic and $T[V \! \setminus \! \setcon{v}]$ evaluates to $1$.
A proof of Lemma \ref{Lemma:Recursion} is given in Appendix \ref{Section:ProofsFramework}.

With the recursion at hand we can evaluate the table $T$ by a dynamic programming.
To this end, we store already computed entries and look them up when needed.
An entry $T[V]$ is evaluated as follows.
We branch over all write events $v \in V$ and test whether the coherence graphs $\graph{loc}[V \! \setminus \! \setcon{v},v]$ and $\graph{mm}[V \! \setminus \! \setcon{v},v]$ are acyclic.
Then, we look up whether $T[V \! \setminus \setcon{v}] = 1$.
If all three queries are positive, we store $T[V] = 1$.
Otherwise, $T[V] = 0$.

The complexity estimation of Theorem \ref{Theorem:MMConsistency} is obtained as follows.
The table has $2^k$ many entries that we evaluate, which constitutes the exponential factor.
For each entry $T[V]$, we branch over at most $k$ write events $v \in V$.
Looking up the value of $T[V \! \setminus \! \setcon{v}]$ can be done in constant time.
The following lemma shows that $\bigO(k \cdot n^2)$ time suffices to construct the coherence graphs and to check them for acyclicity.
The latter checks are based on Kahn's algorithm \cite{Kahn1962} for finding a topological sorting.
This completes the proof of Theorem~\ref{Theorem:MMConsistency}.
\begin{lemma}
	\label{Lemma:CoherenceAcyclicity}
	Let $V \subseteq \WR$ and $v \in \overline{V}$.
	Constructing the coherence graphs $\graph{loc}[V,v]$ and $\graph{mm}[V,v]$ and testing both for acyclicity can be done in time $\bigO(k \cdot n^2)$.
\end{lemma}

\section{Instances of the Framework}
\label{Section:Instantiations}

We show the applicability of our framework and obtain consistency algorithms for the memory models $\SC$, $\TSO$, $\PSO$, and $\RMO$.
To this end, we first need to show that our notion of consistency coincides with the notion of consistency used in the literature for these models.
This ensures that the obtained algorithms really solve the correct problem.
Once this is achieved, we can directly apply the framework to $\SC$, $\TSO$, and $\PSO$. For $\RMO$, we show how the framework can be slightly modified to also capture this more relaxed model.

\subsection{Validity}
\label{Section:Validity}

Consistency, as it is considered in the literature, is also known as \emph{validity} \cite{Alglave2012,Alglave2014}.
We use the latter name to avoid confusion with our notion of consistency.
Before we show that both notions actually coincide, we formally define validity.
The definition is based on \emph{store orders}~\cite{Alglave2012,Alglave2014,Bouajjani2019} (also known as \emph{coherence orders}).
Given a history $h = \history$, a \emph{store order} $\ww \subseteq \WR \! \times \! \WR$ takes the form $\ww = \bigcup_{x \in \Var} \ww_x$ so that each $\ww_x$ is a strict total order on $\WR(x)$.
Phrased differently, store orders are unions of total orders on writes to the same variable.
Note that, in contrast to a total order on $\WR$, a store order does not have any edge between write events referring to distinct variables.

Validity is similar to consistency.
But instead of a total order, the acyclicity requirements need to be satisfied by a store order.
Let $\MM$ be a memory model described by $(\pomm,\rfmm)$.
A history $h = \history$ is $\MM$\emph{-valid} if there exists a store order so that
\begin{align*}
	\graph{loc}^{\ww} = (O, \poloc \cup \rf \cup \ww \cup \fr) 
	\ \ \text{and} \ \ 
	\graph{mm}^{\ww} = (O, \pomm \cup \rfmm \cup \ww \cup \fr)
\end{align*}
are acyclic.
The \emph{from-read} relation is defined by $\fr = \rf^{-1} \circ ww$.
Note that the definition, as in the case of consistency above, omits checking for out of thin air values.
We will later add an explicit test for memory models that require it.
This will not affect the complexity.

We show the equivalence of validity and consistency.
To this end, we need to prove that a store order can be replaced by a total order on the write events while acyclicity is preserved.
The following lemma states the result.
It is crucial for the applicability of our framework.
\begin{lemma}
	\label{Lemma:ValidConsistency}
	A history $h$ is $\MM$-valid if and only if it is $\MM$-consistent.
\end{lemma}

Before we give the proof of Lemma \ref{Lemma:ValidConsistency}, we need an auxiliary statement.
It shows that a store order $\ww$ in $\graph{loc}^\ww$ can be replaced by any linearization of $\ww$ without affecting acyclicity.
Phrased differently, any total order $\tw$ on the write events that contains $\ww$ can be inserted into the graph $\graph{loc}^\ww$ - it will still be acyclic.
We state the corresponding lemma.
\begin{lemma}
	\label{Lemma:GlocLinearization}
	Let $h = \history$ be a history, $\ww$ a store order, and $\tw$ a total order on $\WR$ such that $\ww \subseteq \tw$.
	If $\graph{loc}^\ww$ is acyclic, then so is $\graph{loc}^\tw = (O, \poloc \cup \rf \cup \tw \cup \fr)$.
\end{lemma}

The proof of Lemma \ref{Lemma:GlocLinearization} is given in Appendix \ref{Section:ProofsInstantiations}.
We turn to the proof of Lemma \ref{Lemma:ValidConsistency}.
\begin{proof}[Proof of Lemma \ref{Lemma:ValidConsistency}]
	First assume that $h = \history$ is $\MM$-valid.
	Then there is a store order $\ww$ such that $\graph{loc}^\ww$ and $\graph{mm}^\ww$ are acyclic.
	Consider the edges of the latter graph.
	They form a relation $\ordmm = \pomm \cup \rfmm \cup \ww \cup \fr$.
	Since $\graph{mm}^\ww$ is acyclic, the transitive closure $\ordmm^+$ is a strict partial order on $O$.
	Hence, there exists a linear extension, a strict total order $L$ containing $\ordmm^+$.
	We define $\tw = L \cap \WR \! \times \! \WR$.
	Then, $\tw$ is a total order on $\WR$ and we have $\ww \subseteq L \cap \WR \! \times \! \WR = \tw$.
	We show that $\graph{loc}$ and $\graph{mm}$ are acyclic.
	Note that the latter refer to the graphs from the definition of consistency.

	The store order $\ww$ is contained in $\tw$.
	Hence, we obtain that $\ww_x \subseteq \tw_x$ for each variable $x \in \Var$.
	This implies that $\ww_x = \tw_x$ since $\ww_x$ is total on $\WR(x)$.
	We can deduce $\ww = \bigcup_{x \in \Var} \ww_x = \bigcup_{x \in \Var} \tw_x$ and thus $\cf = \rf^{-1} \circ \bigcup_{x \in \Var} \tw_x = \rf^{-1} \circ \ww = \fr$.
	
	Since $\fr = \cf$, we get the acyclicity of $\graph{loc} = \graph{loc}^\tw$ from Lemma \ref{Lemma:GlocLinearization}.
	The acyclicity of $\graph{mm}$ follows since its edges $\pomm \cup \rfmm \cup \tw \cup \cf$ form a subrelation of $L$.
	A cycle would mean that $L$ has a reflexive element, but $L$ is a strict order.
	Hence, $h$ is $\MM$-consistent.
	
	For the other direction, assume that $h$ is $\MM$-consistent.
	By definition, there is a total order $\tw$ on $\WR$ such that $\graph{loc}$ and $\graph{mm}$ are acyclic.
	We construct the store order $\ww = \bigcup_{x \in \Var} \tw_x$.
	Note that, since $\tw_x$ is total on $\WR(x)$, $\ww$ is indeed a store order and we have $\ww \subseteq \tw$.
	We show that $\graph{loc}^\ww$ and $\graph{mm}^\ww$ are acyclic.
	In fact, we have that $\fr = \rf^{-1} \circ \ww = \cf$.
	This implies that $\graph{loc}^\ww$ and $\graph{mm}^\ww$ are subgraphs of $\graph{loc}$ and $\graph{mm}$, respectively.
	Hence, the two graphs are acyclic and $h$ is $\MM$-valid.
	\qedhere
\end{proof}

\subsection{Instances}
\label{Section:Instances}

We apply the algorithmic framework to the mentioned memory models and obtain (optimal) deterministic algorithms for their corresponding validity/consistency problem.
To this end, we employ the formal description of these models given in \cite{Alglave2012,Alglave2014}.

\subparagraph*{Sequential Consistency.}
\emph{Sequential Consistency} ($\SC$) is a basic memory model, first defined by Lamport in \cite{Lamport1979}.
Intuitively, $\SC$ strictly follows the given program order and flushes each issued write immediately to the memory so that it is visible to all other threads.

Formally, $\SC$ is described by the tuple $\SC = (\posc,\rfsc)$ with $\posc = \po$ and $\rfsc = \rf$.
Hence, it employs the full program order and reads-from relation, making the uniprocessor test on $\graph{loc}$ obsolete.
However, our framework still applies. 
It yields an algorithm for the corresponding validity/consistency problem running in time $\bigO(2^k \cdot k^2 \cdot n^2)$.
We show in Section \ref{Section:LowerBounds} that the obtained algorithm is optimal under $\ETH$.

\subparagraph*{Total Store Ordering.}
The SPARC memory model \emph{Total Store Order} ($\TSO$) \cite{Sparc} resembles a more relaxed memory behavior.
Instead of flushing writes immediately to the memory, like in $\SC$, each thread has an own FIFO buffer and issued writes of that thread are pushed into the buffer.
Writes in the buffer are only visible to the owning thread.
If the owner reads a certain variable, it first looks through the buffer and reads the latest issued write on that variable.
This is called \emph{early read}.
At some nondeterministic point, the buffer is flushed to the memory, making the writes visible to other threads as well.

The formal description of $\TSO$ is given by the tuple $\TSO = (\potso,\rftso)$, where $\potso = \po \! \setminus \! \WR \! \times \! \RD$ is a relaxation of the program order, containing no write-read pairs.
The relation $\rftso = \rf_e$ is a restriction of $\rf$ to write-read pairs from different threads:
\begin{align*}
	\rf_e = \Setcon{(w,r) \in \rf}{(w,r) \notin \po, (r,w) \notin \po}.
\end{align*}
Unlike in the case of $\SC$, we do not have the full program order and reads-from relation at hand.
Hence, the uniprocessor test is essential.
Applying the framework yields an algorithm for the validity/consistency problem of $\TSO$ running in time $\bigO(2^k \cdot k^2 \cdot n^2)$.
The optimality of the obtained algorithm is shown in Section \ref{Section:LowerBounds}.

\subparagraph*{Partial Store Ordering.}
The second SPARC model that we consider is \emph{Partial Store Order} ($\PSO$) \cite{Sparc}.
It is weaker than $\TSO$ since writes to different locations issued by a thread may not arrive at the memory in program order.
Intuitively, in $\PSO$ each thread has a buffer per variable where the corresponding writes to the variable are pushed.
Like for $\TSO$, threads can read early from their buffers and the buffers are, at some point, flushed to the memory.

Formally, $\PSO$ is captured by the tuple $\PSO = (\popso,\rfpso)$.
Here, the relation $\popso = \po \! \setminus \! ( \WR \! \times \! \RD \cup \WR \! \times \! \WR )$ takes away the write-read pairs and the write-write pairs from the program order and, like for $\TSO$, we have $\rfpso = \rf_e$.
Hence, we can apply our framework and obtain an $\bigO(2^k \cdot k^2 \cdot n^2)$-time algorithm.
The obtained algorithm is optimal.

\subparagraph*{Relaxed Memory Order.}
We extend the framework to also capture SPARC's \emph{Relaxed Memory Order} ($\RMO$) \cite{Sparc}.
The model needs an explicit out of thin air test and allows for so-called \emph{load-load hazards}.
We show how both modifications can be built into the framework without affecting the complexity of the resulting consistency algorithm.

The model $\RMO$ relies on an additional \emph{dependency relation} resembling address and data dependencies among events in an execution of a program.
For instance, if a read event has influence on the value written by a subsequent write event.
We assume that the \emph{dependency relation} $\dep$ is given along with a history $h = \history$ and is a subrelation of $\po \cap (\RD \! \times \! O)$.
The latter means that $\dep$ always starts in a read event.
With the relation at hand we can perform an out of thin air test.
In fact, such a test \cite{Alglave2012} requires that $(O,\dep \cup \rf)$ is acyclic.
This can be checked by Kahn's algorithm \cite{Kahn1962} in time $\bigO(n^2)$. 
Hence, the test can be added to the framework without increasing the time complexity of the obtained consistency algorithm.

Load-load hazards are allowed by $\RMO$.
These occur when two reads of the same variable are scheduled not following the program order.
To obtain an algorithm from the framework in this case, we need to weaken the uniprocessor check \cite{Alglave2012}.
In fact, we replace the relation $\poloc$ by $\poloc_\llh = \poloc \setminus \RD \times \RD$ and require that the graph $\graph{loc-llh} = (O, \poloc_\llh \cup \rf \cup \tw \cup \cf)$ is acyclic.
The correctness of the framework is ensured since Lemma \ref{Lemma:ValidConsistency} still holds in this setting.
Moreover, the running time of the resulting algorithm is not affected.

With these modifications, we can obtain a consistency algorithm for $\RMO$.
Formally, $\RMO = (\pormo,\rfrmo)$ where $\pormo = \dep$ and $\rfrmo = \rf_e$.
Applying the framework with out of thin air test and $\graph{loc-llh}$ yields a consistency algorithm running in $\bigO(2^k \cdot k^2 \cdot n^2)$.

\section{Lower Bounds}
\label{Section:LowerBounds}

We show that the framework provides optimal consistency algorithms for $\SC$, $\TSO$, and $\PSO$.
To this end, we employ the $\ETH$ and prove that checking consistency under these three memory models cannot be achieved in subexponential time $2^{o(k)}$.
Since the algorithms obtained in Section \ref{Section:Instantiations} match the lower bound, they are indeed optimal.

We begin with the lower bound for $\SC$-Consistency.
For its proof, we rely on a characterization of the $\ETH$, known as the \emph{Sparsification Lemma} \cite{Impagliazzo2001}.
It states that $\ETH$ is equivalent to the assumption that $\kSAT{3}$ cannot be solved in time $2^{o(n+m)}$, where $n$ is the number of variables and $m$ is the number of clauses of the input formula.
To transport the lower bound to consistency checking, we construct a polynomial-time reduction from $\kSAT{3}$ to $\SC$-Consistency which controls the number of writes $k$.
Technically, for a given formula $\varphi$, the reduction yields a history $h_\varphi$ that has only $k = \bigO(n+m)$ many write events and is $\SC$-consistent if and only if $\varphi$ is satisfiable.
By invoking the reduction, an $2^{o(k)}$-time algorithm for $\SC$-Consistency, would yield an $2^{o(n+m)}$-time algorithm for $\kSAT{3}$, contradicting the $\ETH$.
\begin{theorem}
	\label{Theorem:SCLower}
	$\SC$-Consistency cannot be solved in time $2^{o(k)}$ unless $\ETH$ fails.
\end{theorem}

It is left to construct the reduction.
Let $\varphi$ be a $\kSAT{3}$-instance over the variables $X = \setcon{x_1, \dots, x_n}$ and with clauses $C_1, \dots, C_m$.
Moreover, let $L$ denote the set of literals.
We construct a history $h_\varphi$ the number of writes of which depends linearly on $n+m$.

The main idea of the reduction is to mimic an evaluation of $\varphi$ by an interleaving of the events in $h_\varphi$.
To this end, we divide evaluating $\varphi$ into three steps: (1) choose an evaluation of the variables, (2) evaluate the literals accordingly, and (3) check whether the clauses are satisfied.
For each of these steps we have separate threads taking care of the task.
Scheduling them in different orders will yield different evaluations.
An overview is given in Figure \ref{Figure:LowerBoundSC}.
\begin{figure}[h]
	\begin{tikzpicture}
	\node (Tx0) {$T_0(x):$};
	\node [below = -0.1cm of Tx0] (wrx0) {$\wri(x,0)$};
	
	\node [right = 0.35cm of Tx0] (Tx1) {$T_1(x):$};
	\node [below = -0.1cm of Tx1] (wrx1) {$\wri(x,1)$};
	
	
	\node [right = 1.5cm of Tx1] (Tell0) {$T_0(\ell):$};
	\node [below = -0.1cm of Tell0] (Tell0rdx0a) {$\rd(x,0)$};
	\node [below = -0.1cm of Tell0rdx0a] (Tell0wrell) {$\wri(\ell,c)$};
	\node [below = -0.1cm of Tell0wrell] (Tell0rdx0b) {$\rd(x,0)$};
	
	\node [right = 0.35cm of Tell0] (Tell1) {$T_1(\ell):$};
	\node [below = -0.1cm of Tell1] (Tell1rdx1a) {$\rd(x,1)$};
	\node [below = -0.1cm of Tell1rdx1a] (Tell1wrell) {$\wri(\ell,d)$};
	\node [below = -0.1cm of Tell1wrell] (Tell1rdx1b) {$\rd(x,1)$};
	
	
	\node [right = 1.5cm of Tell1] (TC1) {$T^1(C):$};
	\node [below = -0.1cm of TC1] (TC1rda) {$\rd(\ell_3,0)$};
	\node [below = -0.1cm of TC1rda] (TC1rdb) {$\rd(\ell_1,1)$};
	
	\node [right = 0.35cm of TC1] (TC2) {$T^2(C):$};
	\node [below = -0.1cm of TC2] (TC2rda) {$\rd(\ell_1,0)$};
	\node [below = -0.1cm of TC2rda] (TC2rdb) {$\rd(\ell_2,1)$};
	
	\node [right = 0.35cm of TC2] (TC3) {$T^3(C):$};
	\node [below = -0.1cm of TC3] (TC3rda) {$\rd(\ell_2,0)$};
	\node [below = -0.1cm of TC3rda] (TC3rdb) {$\rd(\ell_3,1)$};
\end{tikzpicture}
	\captionof{figure}{Parts of the history $h_\varphi$ for a variable $x \in X$, a literal $\ell \in L$, and a clause $C = \ell_1 \vee \ell_2 \vee \ell_3$.
	Values of $c$ and $d$ depend on $\ell$.
	If $\ell = x$, then $c = 0, d = 1$. Otherwise, $c = 1, d = 0$.}
	\label{Figure:LowerBoundSC}
\end{figure}

Figure \ref{Figure:LowerBoundSC} presents $h_\varphi$ as a collection of threads.
The program order is obtained from reading threads top to bottom.
The reads-from relation is given since each value is written at most once to a variable.
Hence, there is always a unique write event providing the read~value.

We elaborate on the details of the reduction.
For realizing Step (1), we construct two threads, $T_0(x)$ and $T_1(x)$, for each variable $x \in X$.
These mimic an evaluation of the variable and consist of only one write event.
Thread $T_0(x)$ writes $0$ to $x$, thread $T_1(x)$ writes $1$.
If $T_0(x)$ gets scheduled before $T_1(x)$, variable $x$ is evaluated to $1$ and to $0$ otherwise.
Hence, the thread that is scheduled later will determine the actual evaluation of $x$.

In Step (2), we propagate the evaluation of the variables to the literals.
To this end, we construct two threads for each literal $\ell \in L$.
Let $\ell = x / \neg x$ be a literal on variable $x \in X$.
The first thread $T_0(\ell)$ is responsible for evaluating $\ell$ when $x$ is evaluated to $0$.
It first performs a read event $\rd(x,0)$, followed by $\wri(\ell,c)$ and $\rd(x,0)$.
The value $c$ depends on the literal: if $\ell = x$, then $c = 0$.
Otherwise $c = 1$.
Note that the read events guard the write event.
This ensures that $T_0(\ell)$ can only run if $x$ is already evaluated to $0$ and once $T_0(\ell)$ is running, the evaluation of $x$ cannot change until the thread finishes.
Thread $T_1(\ell)$ behaves similar.
It evaluates the literal $\ell$ when $x$ is evaluated to $1$.
Both threads cannot interfere.
Like for the variables, the later scheduled thread determines the actual evaluation of the literal.

It is left to evaluate the clauses.
For a clause $C = \ell_1 \vee \ell_2 \vee \ell_3$, we have threads $T^1(C)$, $T^2(C)$, and $T^3(C)$ as shown in Figure \ref{Figure:LowerBoundSC}.
It is the task of these threads to ensure that at least one literal in $C$ evaluates to $1$.
To see this, assume we have the contrary, an evaluation of the variables (and the literals) such that $\ell_1$, $\ell_2$, and $\ell_3$ evaluate to $0$.
Due to the construction, $\ell_1$ storing $0$ implies that $\wri(\ell_1,1)$ preceded the write event $\wri(\ell_1,0)$.
Hence, the read event $\rd(\ell_1,1)$ in $T^1(C)$ must have already been scheduled.
In particular, it has to occur before $\rd(\ell_1,0)$ in $T^2(C)$.
Since $\ell_2$ and $\ell_3$ also store $0$, we get a similar dependency among their reads: $\rd(\ell_2,1)$ occurs before $\rd(\ell_2,0)$ and $\rd(\ell_3,1)$ occurs before $\rd(\ell_3,0)$.
Due to program order, we obtain a dependency cycle involving all these reads:
\begin{align*}
	\rd(\ell_1,1) \rightarrow \rd(\ell_1,0) \rightarrow \rd(\ell_2,1) \rightarrow \rd(\ell_2,0) \rightarrow \rd(\ell_3,1) \rightarrow \rd(\ell_3,0) \rightarrow \rd(\ell_1,1).
\end{align*}
An arrow $r \rightarrow r'$ means that $r$ has to precede $r'$ in an interleaving of the events in $h_\varphi$.
Since cycles cannot occur in an interleaving, the threads can only be scheduled properly when a satisfying assignment is given.
The construction of a proper schedule is subtle.
We provide details in Appendix \ref{Section:ProofsLowerBound}.
The following lemma states the correctness of the construction.
\begin{lemma}
	\label{Lemma:SCLowerCorrectness}
	Formula $\varphi$ is satisfiable if and only if the history $h_\varphi$ is $\SC$-consistent.
\end{lemma}

Clearly, $h_\varphi$ can be constructed in polynomial time.
We determine the number of write events.
For each variable $x \in X$ and each literal $\ell \in L$, we introduce two write events.
Hence, $k = 2 \cdot n + 2 \cdot \abs{L}$.
Since there are at most $3 \cdot m$ many literals in $\varphi$, we get that $k$ is bounded by $2 \cdot n + 6 \cdot m$, a number linear in $n+m$.
This finishes the proof of Theorem \ref{Theorem:SCLower}.

We obtain lower bounds for $\TSO$ and $\PSO$, by constructing a similar reduction from $\kSAT{3}$ to $\TSO$ and $\PSO$-Consistency.
To this end, we extend the above reduction by only adding read events that enforce sequential behavior.
Intuitively, we can force the FIFO buffers of $\TSO$ and $\PSO$ to push each issued write to the memory immediately.
Then, the above correctness argument still applies.
The number of write events does not change and is still linear in $n+m$.
This yields the following result.
Details are given in Appendix \ref{Section:ProofsLowerBound}.
\begin{theorem}\label{Thorem:TSOPSOLower}
	$\TSO$ and $\PSO$-Consistency cannot be solved in time $2^{o(k)}$ unless $\ETH$ fails.
\end{theorem}

\bibliographystyle{plain}
\bibliography{content/cite}

\newpage
\appendix
\section{Proofs of Section \ref{Section:Framework}}
\label{Section:ProofsFramework}
\begin{proof}[Proof of Lemma \ref{Lemma:Recursion}]
	Let $V \subseteq \WR$ be non-empty.
	We have to prove two directions.
	To this end, first assume $T[V] = 1$.
	We show that there is an element $v \in V$ such that $\graph{loc}[V \! \setminus \! \setcon{v},v]$ and $\graph{mm}[V \! \setminus \! \setcon{v},v]$ are both acyclic and $T[V \! \setminus \! \setcon{v},v] = 1$.
	
	Since $T[V] = 1$, there is a snapshot ordering $\tw[V] = \torder[V] \cup \rorder[V]$ with a total order $\torder[V]$ on $V$.
	Moreover, the snapshot order satisfies that the graphs
	\begin{align*}
		\graph{loc}(\tw[V]) &= (O, \poloc \cup \rf \cup \tw[V] \cup \cf[V]), \\
		\graph{mm}(\tw[V]) &= (O, \pomm \cup \rfmm \cup \tw[V] \cup \cf[V])
	\end{align*}
	are both acyclic.
	
	We extract the suitable write event.
	Since $\torder[V]$ is total on $V$, there is a unique minimal element $v$ according to the order.
	We set $V' = V \! \setminus \! \setcon{v}$ and show the following three facts:
	
	\begin{enumerate}[(1)]
		\item \label{Proof:F1}$\graph{loc}[V',v]$ is a subgraph of $\graph{loc}(\tw[V])$,
		\item \label{Proof:F2}$\graph{mm}[V',v]$ is a subgraph of $\graph{mm}(\tw[V])$, and
		\item \label{Proof:F3}$T[V'] = 1$.
	\end{enumerate}

	With the these facts at hand, we can conclude that $v$ is the element we were looking. Since $\graph{loc}(\tw[V])$ and $\graph{mm}(\tw[V])$ are acyclic, any subgraph of these are as well.
	
	We begin by proving \textbf{(\ref{Proof:F1})}.
	To this end, we show that each edge of the coherence graph
	\begin{align*}
		\graph{loc}[V',v] = (O, \poloc \cup \rf \cup \rorder[V',v] \cup \cf[V',v])
	\end{align*}
	is already present in $\graph{loc}(\tw[V])$.
	Since $\poloc$ and $\rf$ are already present in $\graph{loc}(\tw[V])$, we need to show that the edges of $\rorder[V',v]$ and $\cf[V',v]$ are also there.
	
	By definition, $\rorder[V',v] = \rorder[V] \cup \Setcon{(v,w)}{w \in V'}$.
	Since $v$ was selected to be the minimal element of $\torder[V]$ on $V$ and $\torder[V]$ is total, we get that each edge $(v,w)$ with $w \in V'$ is also contained in $\torder[V]$.
	Hence, we can deduce the following:
	\begin{align*}
		\rorder[V',v] \subseteq \rorder[V] \cup \torder[V] = \tw[V].
	\end{align*}
	For the edges of $\cf[V',v]$ we then obtain:
	\begin{align*}
		\cf[V',v] = \rf^{-1} \circ \bigcup_{x \in \Var} \rorder[V',v]_x \subseteq \rf^{-1} \circ \bigcup_{x \in \Var} \tw[V]_x = \cf[V],
	\end{align*}
	showing that $\graph{loc}[V',v]$ is a subgraph of $\graph{loc}(\tw[V])$.
	
	The proof of \textbf{(\ref{Proof:F2})} follows from \textbf{(\ref{Proof:F1})}.
	We have to show that each edge of 
	\begin{align*}
		\graph{mm}[V',v] = (O, \pomm \cup \rfmm \cup \rorder[V',v] \cup \cf[V',v])
	\end{align*}
	is contained in $\graph{mm}(\tw[V])$.
	The edges of $\pomm$ and $\rfmm$ are already present in $\graph{mm}(\tw[V])$. 
	Since $\rorder[V',v] \subseteq \tw[V]$ and $\cf[V',v] \subseteq \cf[V]$ hold by \textbf{(\ref{Proof:F1})}, we get that $\graph{mm}[V',v]$ is a proper subgraph of $\graph{mm}(\tw[V])$.
	
	It is left to prove \textbf{(\ref{Proof:F3})}.
	To this end, we construct a snapshot order $\tw[V']$ on $V'$ such that $\graph{loc}(\tw[V'])$ is a subgraph of $\graph{loc}(\tw[V])$ and $\graph{mm}(\tw[V'])$ is a subgraph of $\graph{mm}(\tw[V])$.
	This shows that $T[V'] = 1$ since the two latter graphs are acyclic.
	
	We construct the snapshot order $\tw[V']$ as follows.
	Set $\tw[V'] = \torder[V'] \cup \rorder[V']$, where $\rorder[V'] = \Setcon{(\overline{w},w)}{\overline{w} \in \overline{V'}, w \in V'}$ and $\torder[V'] = \torder[V] \cap (V' \! \times V')$ is the restriction of $\torder[V]$ to the set $V'$.
	Note that $\torder[V']$ is total on $V'$.
	Hence, $\tw[V']$ is a proper snapshot order.
	
	By definition we get that $\torder[V'] \subseteq \torder[V]$.
	Now consider an edge $(\overline{w},w)$ from $\rorder[V']$ with $\overline{w} \in \overline{V'}$ and $w \in V'$.
	There are two cases:
	(1) For $\overline{w} = v$, the edge $(v,w)$ is already contained in $\torder[V]$ since $v$ was chosen to be $\torder[V]$-minimal and $\torder[V]$ is total on $V$.
	(2) For $\overline{w} \neq v$, we get that $\overline{w} \in \overline{V}$.
	Hence, the edge $(\overline{w},w)$ is already contained in $\rorder[V]$.
	Putting the cases together, we obtain the following inclusions:
	\begin{align*}
		\tw[V'] &= \torder[V'] \cup \rorder[V'] \subseteq \torder[V] \cup \rorder[V] = \tw[V], \\
		\cf[V'] &= \rf^{-1} \circ \bigcup_{x \in \Var} \tw[V']_x \subseteq \rf^{-1} \circ \bigcup_{x \in \Var} \tw[V]_x = \cf[V].
	\end{align*}
	
	From these inclusions, we immediately obtain that $\graph{loc}(\tw[V'])$ is a subgraph of $\graph{loc}(\tw[V])$ and that $\graph{mm}(\tw[V'])$ is a subgraph of $\graph{mm}(\tw[V])$.
	
	For the other direction assume the existence a write event $v \in V$ such that the coherence graphs $\graph{loc}[V',v]$ and $\graph{mm}[V',v]$ are acyclic and $T[V'] = 1$.
	Here, $V' = V \! \setminus \! \setcon{v}$.
	In order to show that $T[V] = 1$, we need to construct a snapshot order $\tw[V]$ on $V$ such that the graphs $\graph{loc}(\tw[V])$ and $\graph{mm}(\tw[V])$ are acyclic.
	
	By the assumption $T[V'] = 1$, there is a snapshot order $\tw[V'] = \torder[V'] \cup \rorder[V']$ such that
	\begin{align*}
		\graph{loc}(\tw[V']) &= (O,\poloc \cup \rf \cup \tw[V'] \cup \cf[V']), \\
		\graph{mm}(\tw[V']) &= (O, \poloc \cup \rfmm \cup \tw[V'] \cup \cf[V'])
	\end{align*}
	are both acyclic.
	We extend the order $\torder[V']$ by adding $v$ as new minimal element.
	Define $\torder[V] = \torder[V'] \cup \Setcon{(v,w)}{w \in V'}$.
	Then, $\torder[V]$ is a total order on $V$.
	Thus, $\tw[V] = \torder[V] \cup \rorder[V]$ with relation $\rorder[V] = \Setcon{(\overline{w},w)}{\overline{w} \in \overline{V}, w \in V}$ is a snapshot order on $V$.
	
	We show the acyclicity of $\graph{loc}(\tw[V])$ and $\graph{mm}(\tw[V])$ in two steps.
	First, we define intermediary graphs $\Jgraph{loc}[V,v]$ and $\Jgraph{mm}[V,v]$ and show that $\graph{loc}(\tw[V])$ and $\graph{mm}(\tw[V])$ are subgraphs.
	In the second step we prove that $\Jgraph{loc}[V,v]$ and $\Jgraph{mm}[V,v]$ are acyclic.
	In fact, we show that a cycle in one of the two graphs would induce a cycle in one of the coherence graphs $\graph{loc}[V',v]$ and $\graph{mm}[V',v]$ or in $\graph{loc}(\tw[V'])$ and $\graph{mm}(\tw[V'])$, which are acyclic by assumption.
	The acyclicity of $\graph{loc}(\tw[V])$ and $\graph{mm}(\tw[V])$ follows and we obtain $T[V] = 1$.
	
	We begin with the first step.
	Before we define the intermediary graphs, we need two new relations depending on the fact that $v$ is the new minimal element of $V$.
	Define
	\begin{align*}
		\inc(v) = \Setcon{(\overline{w},v)}{\overline{w} \in \overline{V}} ~\text{and}~
		\cfinc(v) = \rf^{-1} \circ \bigcup_{x \in \Var} \inc(v)_x.
	\end{align*}
	Note that a pair $(r,v)$ is in $\cfinc(v)$ if $r$ is a read event with a write event $\overline{w} \in \overline{V}$ such that $(\overline{w},r) \in \rf$, $(\overline{w},v) \in \inc(v)$, and $\overline{w}$ and $v$ write to the same variable $x$.
	
	The graphs $\Jgraph{loc}[V,v]$ and $\Jgraph{mm}[V,v]$ are now defined as follows:
	\begin{align*}
		\Jgraph{loc}[V,v] &= (O, \poloc \cup \rf \cup \tw[V'] \cup \inc(v) \cup \cf[V'] \cup \cfinc(v)), \\
		\Jgraph{mm}[V,v] &= (O, \pomm \cup \rfmm \cup \tw[V'] \cup \inc(v) \cup \cf[V'] \cup \cfinc(v)).
	\end{align*}
	
	We show that $\graph{loc}(\tw[V])$ is a subgraph of $\Jgraph{loc}[V,v]$.
	First note that the edges of $\poloc$ and $\rf$ are already present in $\Jgraph{loc}[V,v]$.
	It is left to argue that $\tw[V]$ and $\cf[V]$ are included in the edges of $\Jgraph{loc}[V,v]$ as well.
	To this end, consider the following inclusion:
	\begin{align*}
		\torder[V] = \torder[V'] \cup \Setcon{(v,w)}{w \in V'}
		\subseteq \torder[V'] \cup \rorder[V'] = \tw[V'].
	\end{align*}
	The first equality is the definition of $\torder[V]$.
	The inclusion holds since $v \in \overline{V'}$ and thus $\Setcon{(v,w)}{w \in V'} \subseteq \rorder[V']$.
	Relation $\rorder[V]$ is embedded as follows:
	\begin{align*}
		\rorder[V] &= \Setcon{(\overline{w},w)}{\overline{w} \in \overline{V}, w \in V} \\
		&= \Setcon{(\overline{w},w)}{\overline{w} \in \overline{V}, w \in V'} \cup \Setcon{(\overline{w},v)}{\overline{w} \in \overline{V}} \\
		&\subseteq \rorder[V'] \cup \inc(v).
	\end{align*}
	The latter inclusion holds due to the fact that $\overline{V} \subseteq \overline{V'}$.
	Combining the above inclusions then yields $\tw[V] \subseteq \tw[V'] \cup \inc(v)$.
	For the conflict relation, we consequently obtain
	\begin{align*}
		\cf[V] &= \rf^{-1} \circ \bigcup_{x \in \Var} \tw[V]_x \\
		&\subseteq \rf^{-1} \circ \bigcup_{x \in \Var} \left( \tw[V'] \cup \inc(v) \right)_x \\
		&= \rf^{-1} \circ \bigcup_{x \in \Var} \left( \tw[V']_x \cup \inc(v)_x \right) \\
		&=  \left( \rf^{-1} \circ \bigcup_{x \in \Var} \tw[V']_x \right) \cup 
		\left( \rf^{-1} \circ \bigcup_{x \in \Var} \inc(v)_x \right) \\
		&= \cf[V'] \cup \cfinc(v).
	\end{align*}
	
	Hence, all edges of $\graph{loc}(\tw[V])$ are present in $\Jgraph{loc}[V,v]$ which proves the subgraph relation.
	
	The fact that $\graph{mm}(\tw[V])$ is a subgraph of $\Jgraph{mm}[V,v]$ follows easily from the above observations. 
	Since the relations $\pomm$ and $\rfmm$ are already present in $\Jgraph{mm}[V,v]$ and $\tw[V] \subseteq \tw[V'] \cup \inc(v)$ as well as $\cf[V] \subseteq \cf[V'] \cup \cfinc(v)$ hold independently from the considered graphs, we obtain the desired subgraph relation.
	
	In the second step, we show the acyclicity of $\Jgraph{loc}[V,v]$ and $\Jgraph{mm}[V,v]$.
	We focus on $\Jgraph{loc}[V,v]$ since the proof for $\Jgraph{mm}[V,v]$ is similar.
	Assume there is a cycle $C$ in $\Jgraph{loc}[V,v]$.
	If $C$ does neither contain an edge from $\inc(v)$ nor from $\cfinc(v)$, the cycle has only edges over $\poloc \cup \rf \cup \tw[V'] \cup \cf[V']$.
	Hence, $C$ is a cycle in $\graph{loc}(\tw[V'])$ which is a contradiction since the graph is acyclic.
	Therefore, $C$ goes through at least one edge from $\inc(v)$ or $\cfinc(v)$.
	In both cases, this means that $C$ passes through the write event $v$.
	We may think of $C$ as a cycle that starts and ends in $v$:
	$C$ is of the form
	\begin{align*}
		C = e_0 . e_1 \dots e_\ell
	\end{align*}
	with $e_i$ edges and $e_0 = (v,w_1)$, $e_\ell = (w_\ell,v)$ for events $w_1,w_\ell \in O$.
	Moreover, we assume that $C$ is short.
	The write event $v$ is only visited once.
	Otherwise, we would get a shorter cycle.
	
	Out of $C$, we show how to construct a cycle $\hat{C}$ in $\graph{loc}[V',v]$ which contradicts the assumption that the coherence graphs are acyclic.
	To this end, we induct over the edges of $C$ and construct $\hat{C}$ while keeping the invariant that all edges of $\hat{C}$ are in $\graph{loc}[V',v]$.
	
	Initially, $\hat{C}$ does not have any edges.
	The induction step is as follows.
	Assume we have already constructed a part of $\hat{C}$ while iterating to the $i$-th edge $e$ of $C$.
	We get the following case distinction, based upon the type of $e$:
	
	\begin{itemize}
		\item If $e$ is an edge in $\poloc \cup \rf$.
		Then $e$ is also present in $\graph{loc}[V',v]$ and we can add it to $\hat{C}$ by setting: $\hat{C} = \hat{C} . e$.
		
		\item If $e$ is an edge in $\tw[V']$ we get two subcases:
		(1) If $e$ is in $\torder[V']$.
		Then, $e = (w,w')$, where $w,w' \in V'$ are write events.
		There is an edge $(v,w')  \in \rorder[V',v]$ by definition.
		We delete the content of $\hat{C}$ and start a new cycle with this edge:
		$\hat{C} = (v,w')$.
		It lies in $\graph{loc}[V',v]$.
		
		(2) If $e$ is an edge in $\rorder[V']$.
		Then $e = (\overline{w},w)$ where $\overline{w} \in \overline{V'}$ and $w \in V'$.
		The edge then also lies in $\rorder[V',v]$ and thus in $\graph{loc}[V',v]$.
		We add it to $\hat{C}$: $\hat{C} = \hat{C}.e$.
		
		\item If $e$ is an edge in $\inc(v)$.
		In this case, $e$ is of the form $(\overline{w},v)$ with $\overline{w} \in \overline{V}$.
		Thus, $e$ lies in $\rorder[V',v]$ and therefore in $\graph{loc}[V',v]$.
		We add the edge to $\hat{C}$ by $\hat{C} = \hat{C}.e$.
		
		\item If $e$ is an edge in $\cf[V']$.
		Then $e = (r,w)$, where $r$ is a read event and $w \in V'$ is a write event.
		There is an edge $(v,w) \in \rorder[V',v]$ and thus in $\graph{loc}[V',v]$.
		We delete $\hat{C}$ and start a new cycle via $\hat{C} = (v,w)$.
		
		\item If $e$ is an edge in $\cfinc(v)$.
		Then $e$ lies in $\cf[V',v]$ since $\inc(v) \subseteq \rorder[V',v]$ and
		\begin{align*}
			\cfinc(v) = \rf^{-1} \circ \bigcup_{x \in \Var} \inc(v)_x \subseteq \rf^{-1} \circ \bigcup_{x \in \Var} \rorder[V',v]_x = \cf[V',v].
		\end{align*}
		Hence, $e$ can be added by $\hat{C} = \hat{C}.e$.
	\end{itemize}

	In the construction, the first edge of $\hat{C}$ always leaves the write event $v$, it is of the form $(v,w)$ for some event $w \in O$.
	Moreover, the edges in $\inc(v)$ and $\cfinc(v)$ are the only edges in $\Jgraph{loc}[V,v]$ that are incoming for $v$.
	Such an edge is always the last edge of $C$ and does never get deleted during the construction of $\hat{C}$.
	Note that we assumed the existence of such an edge.
	Hence, by construction $\hat{C}$ is a non-empty cycle in $\graph{loc}[V',v]$ that starts and ends in $v$.
	This contradicts the acyclicity of the coherence graph.
	Altogether, $\Jgraph{loc}[V,v]$ is acyclic.
	\qedhere
\end{proof}
\begin{proof}[Proof of Lemma \ref{Lemma:CoherenceAcyclicity}]
	We provide a proof for $\graph{loc}[V,v]$ since the statement for $\graph{mm}[V,v]$ is shown similarly.
	First, we focus on the construction of the graph.
	Recall that
	\begin{align*}
		\graph{loc}[V,v] = (O, \poloc \cup \rf \cup \rorder[V,v] \cup \cf[V,v]).
	\end{align*}
	Constructing the vertices can clearly be done in time $\bigO(n)$, as $n = \abs{O}$.
	The edges of $\poloc \cup \rf$ are part of the input.
	Hence, we can iterate over these edges and add them to the graph.
	Since $\poloc \cup \rf$ is a relation in $O \times O$, this takes time at most $\bigO(n^2)$.
	Next, we construct the edges of the relation
	\begin{align*}
		\rorder[V,v] = \Setcon{(\overline{w},w)}{\overline{w} \in \overline{V \cup \setcon{v}}, w \in V \cup \setcon{v}} \cup \Setcon{(v,w)}{w \in V}.
	\end{align*}
	The latter part is simple to construct:
	we add an edge $(v,w)$ for each $w \in V$.
	These are at most $\bigO(k)$ and takes the same amount of time.
	For constructing the former relation, we iterate over $\overline{w} \in \overline{V \cup \setcon{v}}$ and $w \in V \cup \setcon{v}$ and add the edge $(\overline{w},w)$.
	This takes time at most $\bigO(k^2) = \bigO(k \cdot n)$ time.
	Hence, the relation $\rorder[V,v]$ contains at most $\bigO(k \cdot n)$ many edges and can be constructed in time $\bigO(k \cdot n)$.
	
	It is left to construct the conflict relation $\cf[V,v] = \rf^{-1} \circ \bigcup_{x \in \Var} \rorder[V,v]_x$.
	To this end, we first construct the relation $\rf^{-1}$ by turning around the edges stored in $\rf$.
	Note that $\rf$ consists of at most $\bigO(n)$ many of these since it contains exactly one edge for each read event.
	Hence, $\rf^{-1}$ can be constructed in time $\bigO(n)$.
	The relations $\rorder[V,v]_x$ can be constructed from $\rorder[V,v]$. 
	We iterate over the edges in $\rorder[V,v]$ and put an edge $(w,w')$ to the corresponding projection $\rorder[V,v]_x$ if both write events $w,w'$ write to variable $x$.
	This takes time at most $\bigO(k \cdot n)$.
	The composition $\cf[V,v] = \rf^{-1} \circ \bigcup_{x \in \Var} \rorder[V,v]_x$ is then obtained as follows.
	We iterate over all edges $(r,w)$ in $\rf^{-1}$ and $(w',\hat{w})$ in one of the $\rorder[V,v]_x$ and add $(r,\hat{w})$ to $\cf[V,v]$ if $w = w'$.
	Since $\rf^{-1}$ contains at most $\bigO(n)$ many edges and the union of the $\rorder[V,v]_x$ contains at most $\bigO(k \cdot n)$ many edges, constructing $\cf[V,v]$ takes time $\bigO(k \cdot n^2)$.
	
	Hence, the graph $\graph{loc}[V,v]$ can be constructed in time $\bigO(k \cdot n^2)$.
	It is left to show that cycles in $\graph{loc}[V,v]$ can be detected within the same amount of time.
	To this end, we apply Kahn's algorithm \cite{Kahn1962}.
	It finds a topological sorting for a given graph.
	Such a sorting only exists if the graph is acyclic.
	If this is not the case, the algorithm outputs an error.
	Kahn's algorithm runs in time linear in the vertices and edges.
	In our setting, it needs at most $\bigO(n^2)$ time since we can have at most $n^2$ many edges.
	This is below the bound of $\bigO(k \cdot n^2)$ and therefore finishes the proof of the lemma.
	\qedhere
\end{proof}
\section{Proofs of Section \ref{Section:Instantiations}}
\label{Section:ProofsInstantiations}
\begin{proof}[Proof of Lemma \ref{Lemma:GlocLinearization}]
	First, we consider the structure of the acyclic graph $\graph{loc}^\ww$.
	In fact, $\graph{loc}^\ww$ decomposes into a disjoint union of its projections to the variables.
	We show that
	\begin{align*}
		\graph{loc}^\ww = \bigcup_{x \in \Var} \graph{loc}^\ww(x),
	\end{align*}
	where $\graph{loc}^\ww(x) = (O(x), \poloc_x \cup \rf_x \cup \ww_x \cup \fr_x)$ is the projection to variable $x$ and the union is taken over vertices and edges.
	
	It is clear that each projection $\graph{loc}^\ww(x)$ is contained in $\graph{loc}^\ww$ as $O(x) \subseteq O$ and each projected relation is a subset of the original relation.
	For the other inclusion, first note that the set of vertices $O$ is contained in the union since we can write $O = \bigcup_{x \in \Var} O(x)$.
	Phrased differently, each event in $O$ refers to exactly one location.
	It is left to show that all edges of $\graph{loc}^\ww$ are contained in the union.
	By definition, each of the relations $\poloc$, $\rf$, $\ww$, and $\fr$ only relates events to the same location.
	Hence, an edge $(w,w')$ from one of the relations is an edge among events on a variable $x$ and therefore contained in the graph $\graph{loc}^\ww(x)$.
	The union is disjoint since there is no edge in $\graph{loc}^\ww$ that involves events on different variables.
	
	The order $\tw$ contains the store order $\ww$ and for each $x \in \Var$, order $\ww_x$ is total on $\WR(x)$.
	This implies, that $\tw_x = \ww_x$. 
	Hence, $\tw$ differs from $\ww$ by additional edges among write events on different variables.
	Formally, we can write $\tw$ as a disjoint union:
	\begin{align*}
		\tw = \ww \cup \ext,
	\end{align*}
	where $\ext = \Setcon{(w,w') \in \tw}{\var(w) \neq \var(w')}$.
	This means that the graph $\graph{loc}^\tw$ of interest has a structure similar to $\graph{loc}^\ww$ but with edges connecting the projections:
	\begin{align*}
		\graph{loc}^\tw = \ext \cup \bigcup_{x \in \Var} \graph{loc}^\ww(x).
	\end{align*}
	In the union, we interpret the relation $\ext$ as a graph with vertices $O$ and edges $\ext$.
	
	Now assume that there is a cycle $C$ in the graph $\graph{loc}^\tw$.
	Then, $C$ takes the following form:
	\begin{align*}
		C = o_1 \xrightarrow{\pi_1} o'_1 
		\xrightarrow{\ext} o_2
		\xrightarrow{\pi_2} o'_2
		\xrightarrow{\ext} \dots
		\xrightarrow{\ext} o_\ell
		\xrightarrow{\pi_\ell} o'_\ell,
	\end{align*}
	where (1) $o_i,o'_i$ are write events in a graph $\graph{loc}^\ww(x_i)$ for a variable $x_i$, (2) $o'_\ell = o_1$, (3) each $\pi_i$ is a path within $\graph{loc}^\ww(x_i)$, and (4) each edge $o'_i \xrightarrow{\ext} o_{i+1}$ is an edge in $\ext$.
	
	Since $o_i$ and $o'_i$ are write events on the same variable $x_i$ and $\ww_{x_i}$ is a total order on these events, there is a relation between the two writes: either $(o_i,o'_i) \in \ww_{x_i}$ or $(o'_i,o_i) \in \ww_{x_i}$.
	In the latter case, we would immediately get a cycle $o'_i \xrightarrow{\ww_{x_i}} o_i \xrightarrow{\pi_i} o'_i$ in the graph $\graph{loc}^\ww(x_i)$.
	But as a subgraph of $\graph{loc}^\ww$, the graph is acyclic and the cycle cannot appear.
	Hence, we get that $(o_i,o'_i) \in \ww_{x_i}$ for each $i$.
	Since $\ww_{x_i}$ is contained in $\tw$, we have an edge $(o_i,o'_i) \in \tw$ for each $i$.
	Hence, we can shorten the cycle $C$ to a cycle $C^\tw$ of the form:
	\begin{align*}
		C^\tw = o_1 \xrightarrow{\tw} o'_1 \xrightarrow{\tw} o_2 \xrightarrow{\tw} o'_2 \xrightarrow{\tw} \dots \xrightarrow{\tw} o_\ell \xrightarrow{\tw} o'_\ell.
	\end{align*}
	Note that we used the fact $\ext \subseteq \tw$.
	The cycle $C^\tw$ contradicts the fact that $\tw$ is a strict total order of $\WR$.
	Hence, cycle $C$ cannot exist and $\graph{loc}^\tw$ is acyclic.
\end{proof}
\section{Proofs of Section \ref{Section:LowerBounds}}
\label{Section:ProofsLowerBound}
\begin{proof}[Proof of Lemma \ref{Lemma:SCLowerCorrectness}]
	Let $\varphi$ be satisfiable.
	We show that $h_\varphi$ is $\SC$-consistent.
	Since the formula is satisfiable, there is an evaluation function $v : X \rightarrow \setcon{0,1}$ that evaluates $\varphi$ to $1$.
	In order to prove that $h_\varphi$ is $\SC$-consistent, we need to construct a total order $\tw$ on the write events of $h_\varphi$ such that $\graph{sc}$ is acyclic.
	Note that the acyclicity of $\graph{loc}$ is implied.
	In fact, we construct a topological sorting of the vertices of $\graph{sc}$ which implies acyclicity.
	
    For the construction of $\tw$, we first extract an ordering on the literals of each clause.
    For any clause $C_i = \ell^i_1 \vee \ell^i_2 \vee \ell^i_3$, let $L(C_i,1) = \Setcon{\ell^i_j \in C_i}{v(\ell^i_j) = 1}$ be the set of literals in $C_i$ that evaluate to $1$. 
    Since $v$ is satisfying, we get that for each $i \in \setcon{1,\dots,m}$, the set of literals $L(C_i,1)$ is non-empty.
    Similarly, we define $L(C_i,0) = \Setcon{\ell^i_j \in C_i}{v(\ell^i_j) = 0}$.
    
	We begin by constructing a partial order among the literals of the clauses.
    To this end, let $\Nxt(j) = (j \mod 3) + 1$ for $j = 1,2,3$. 
    We first let all literals of a clause that evaluate to $0$ be smaller than the literals evaluating to $1$.
    Formally, we set $L(C_i,0) < L(C_i,1)$ for each $i \in \setcon{1, \dots, m}$. 
    If we find that $\abs{L(C_i,0)} > 1$, there is are two literals evaluating to $0$ (note that it cannot be three).
    In this case, let $L(C_i,0) = \setcon{\ell^i_j, \ell^i_{j'}}$ where $j' = \Nxt(j)$.
    Note that the literals in $L(C_i,0)$ always have this form.
    We set $\ell^i_j < \ell^i_{j'}$.
    The reason why we construct the order like this is that in a topological sorting (linear order) of the events of $h_\varphi$, the corresponding read events of the literals will respect this order.
    
	\subparagraph*{The total order.} 
	We construct $\tw$.
	To this end, we consider the following sets of threads:
	\begin{align*}
		\First(X) &= \Setcon{T_0(x)}{v(x) = 1} \cup \Setcon{T_1(x)}{v(x) = 0}, \\
		\Second(X) &= \Setcon{T_0(x)}{v(x) = 0} \cup \Setcon{T_1(x)}{v(x) = 1}.
	\end{align*}
	The set $\First(X)$ contains those threads that will write the complement evaluation of $v$ into the variables.
	These threads have to run first in an interleaving/topological.
	The variables then get overwritten by the threads of $\Second(X)$.
	These write the correct evaluation $v$ to the variables.
	We need further notation to construct $\tw$.
	Let $\ell$ be a literal over $x \in X$.
	We set
	\begin{align*}
		T_{neg}(\ell) = \left\lbrace
		\begin{aligned}
			T_0(\ell),& ~\text{if}~ \ell = x, \\
			T_1(\ell),& ~\text{otherwise}.
		\end{aligned}
		\right.
	\end{align*}
	This is the thread that stores $0$ in $\ell$. 
	Similarly, we may define a notation for the thread of $h_\varphi$ that stores $1$ in $\ell$.
	It is given by:
	\begin{align*}
		T_{pos}(\ell) = \left\lbrace
		\begin{aligned}
			T_0(\ell),& ~\text{if}~ \ell = \neg x, \\
			T_1(\ell),& ~\text{otherwise}.
		\end{aligned}
		\right.
	\end{align*}
	We go on with the definition of further sets.
	Let $C_i$ be a clause.
	$\First(C_i)$ is the set of threads that write $0$ to the literals that are evaluated to $1$ under $v$.
	It also contains the threads writing $1$ to literals that are evaluated to $0$:
	\begin{align*}
		\First(C_i) = \Setcon{T_{neg}(\ell)}{\ell \in C_i, v(\ell) = 1} \cup \Setcon{T_{pos}(\ell)}{\ell \in C_i, v(\ell) = 0}.
	\end{align*}
	The idea is similar as above.
	In an interleaving of all events, the threads of $\First(C_i)$ run first.
	The variables then get overwritten by the threads of $\Second(C_i)$.
	These forward the evaluation $v$ of the variables to the literals:
	\begin{align*}
		\Second(C_i) = \Setcon{T_{pos}(\ell)}{\ell \in C_i, v(\ell) = 1} \cup \Setcon{T_{neg}(\ell)}{\ell \in C_i, v(\ell) = 0}.
	\end{align*}
	
	The total order $\tw$ consists of several parts obtained from ordering the above sets.
	Let $\Lin_{\WR}(\First(X))$ be some total order on the write events of the threads occurring in $\First(X)$, based on an assumed order on the variables.
	Similarly, let $\Lin_{\WR}(\Second(X))$ be a total order on the write events of the threads in $\Second(X)$.
	Also the second order respects the assumed order on the variables.
	Further, let $\Lin_{\WR}(\First(C_i))$ be a total order on the write events of $\First(C_i)$ that respects the above order on literals (where we see literals as variables of $h_\varphi$).
	Similarly, let $\Lin_{\WR}(\Second(C_i))$ be a total order on write events from the threads of $\Second(C_i)$.
	
	To finally define $\tw$, we use a suitable \emph{append operator}.
	Let $t$ and $r$ be two total orders.
	Then $t.r$ is the total order obtained from ordering the elements of $t$ to be smaller than the elements of $r$ while preserving the orders $t$ and $r$, meaning $t, r \subseteq t.r$.
	The total order $\tw$ on all write events of $h_\varphi$ is then given by combining the total orders defined above as follows.
	Define $\tw = \tw_{first} . \tw_{sec}$, where 
	\begin{align*}
		\tw_{first} &= \Lin_{\WR}(\First(X)) . \Lin_{\WR}(\First(C_1)) \dots \Lin_{\WR}(\First(C_m)), \\
		\tw_{sec} &= \Lin_{\WR}(\Second(X)) . \Lin_{\WR}(\Second(C_1)) \dots \Lin_{\WR}(\Second(C_m)).
	\end{align*}
	
	\subparagraph*{Interleaving the events.}
	We construct an interleaving of all events following the total order $\tw$.
	First, we store the complement evaluation $\bar{v}$ of $v$.
	This is achieved by scheduling $\First(X)$ in the beginning.
	We run these threads in an assumed order on the variables.
	Then, we forward the evaluation $\bar{v}$ to the literals.
	Since the literal threads are guarded by read events, there is a read dependency, meaning that the threads in $\First(X)$ provide the read values for the threads in $\First(C_i)$ for each $i$, see Figure \ref{Figure:LowerBoundSC}.
	For providing the values, we run $\First(X)$ first, followed by $\First(C_i)$ for each $i$.
	Similarly, running $\Second(X)$ will store the evaluation $v$ in the variables and the threads in $\Second(C_i)$ will push it to the literals.
	
	The total order $\tw$ provides the write events in such a way that the read events of the clause threads $T^1(C)$, $T^2(C)$, and $T^3(C)$ can be scheduled properly.
  	Let $C = \ell_1 \vee \ell_2 \vee \ell_3$ be a clause.
  	We distinguish the following cases:
  	
  	(1) If $C$ is satisfied by all three literals, we have $v(\ell_i) = 1$ for $i = 1,2,3$.
  	Then, under $\bar{v}$, we have $\bar{v}(\ell_i) = 0$ for $i = 1,2,3$.
  	Since the write events in $\First(C)$ evaluate the literals under $\bar{v}$, we can schedule the read events $\rd(\ell_1,0)$, $\rd(\ell_2,0)$, and $\rd(\ell_3,0)$ since the values are provided.
  	Technically, we schedule these events after $\tw_{first}$ and before $\tw_{sec}$.
  	The remaining reads $\rd(\ell_1,1)$, $\rd(\ell_2,1)$, and $\rd(\ell_3,1)$ can then be scheduled after $\tw_{sec}$.
  	
  	(2) $C$ is satisfied by two literals.
  	Without loss of generality, we assume $v(\ell_1) = 1$, $v(\ell_2) = 1$, and $v(\ell_3) = 0$.
  	Then $\bar{v}(\ell_1) = 0$, $\bar{v}(\ell_2) = 0$ and $\bar{v}(\ell_3) = 1$.
  	To schedule all the reads of the clause threads properly, we do the following:
  	After $\tw_{first}$, we schedule the reads in the following order $\rd(\ell_1,0) . \rd(\ell_2,0) . \rd(\ell_3,1)$.
  	Note that this conforms to program order and that all the values are provided by the reads in $\tw_{first}$.
  	After $\tw_{sec}$, where the literals are evaluated according to $v$, we can then schedule $\rd(\ell_3,0) . \rd(\ell_1,1) . \rd(\ell_2,1)$.
   
   (3) $C$ is satisfied by one literal.
   Let us assume $v(\ell_1) = 1$, $v(\ell_2) = 0$, and $v(\ell_3) = 0$.
   Then, we get $\bar{v}(\ell_1) = 0$, $\bar{v}(\ell_2) = 1$, and $\bar{v}(\ell_3) = 1$.
   In this case, we schedule the reads $\rd(\ell_1,0) . \rd(\ell_2,1)$ immediately after $\tw_{first}$.
   We cannot schedule $\rd(\ell_3,1)$ since it is blocked by the read $\rd(\ell_2,0)$ the value of which we have not provided yet.
   For providing it, consider $T_{neg}(\ell_2)$.
   The thread occurs in $\Second(C)$ and is scheduled within $\tw_{sec}$.
   Immediately after it performed its write $\wri(\ell_2,0)$, we schedule the read $\rd(\ell_2,0)$ which was blocking.
   Since we did not change the content of variable $\ell_3$ yet, we can schedule $\rd(\ell_3,1)$.
   Note that this fact relies on the order among literals defined above.
   We know that in $\Lin_{\WR}(\Second(C))$, the thread $T_{neg}(\ell_2)$ precedes $T_{neg}(\ell_3)$.
   After the described schedule, $T^2(C)$ and $T^3(C)$ are completely executed.
   After $\tw_{sec}$, the $\ell_i$ store the evaluation under $v$.
   We can then schedule the remaining reads $\rd(\ell_3,0) . \rd(\ell_1,1)$.
   
   By constructing a schedule following these rules for each clause, we obtain an interleaving/total order on all events.
   This implies that $\graph{sc}$ is acyclic.
	
	For the other direction of the proof, assume that $h_\varphi$ is $\SC$-consistent. 
	We show that $\varphi$ is satisfiable.
	By definition, we obtain a total order $\tw$ on the write events of $h_\varphi$ such that $\graph{sc}$ is acyclic.
	We construct the evaluation $v : X \rightarrow \setcon{0,1}$ along the total order:
	\begin{align*}
		v(x) = 1 ~\text{if and only if}~ \wri(x,0) \xrightarrow{\tw} \wri(x,1).
	\end{align*}
	Hence, variable $x$ admits the value that is written latest in $\tw$.
	We show that the evaluation can be consistently extended to the literals.
	For each literal $\ell$ we have:
	\begin{align*}
		v(\ell) = 1 ~\text{if and only if}~ \wri(\ell,0) \xrightarrow{\tw} \wri(\ell,1).
	\end{align*}
	To prove this, let $\ell$ be a literal evaluating to $1$ under $v$.
	Towards a contradiction, suppose that $\wri(\ell,1) \xrightarrow{\tw} \wri(\ell,0)$.
	Without loss of generality, we assume that $\ell = x$.
	The prove for $\ell = \neg x$ is similar.
	Since $v(x) = 1$ as well, we get $\wri(x,0) \xrightarrow{\tw} \wri(x,1)$.
	This will yield the following cycle in $\graph{sc}$:
	\begin{align*}
		\wri(\ell,1) \xrightarrow{\tw} \wri(\ell,0) \xrightarrow{\po} \rd(x,0) \xrightarrow{\cf} \wri(x,1) \xrightarrow{\rf} \rd(x,1) \xrightarrow{\po} \wri(\ell,1).
	\end{align*}
	Hence, we get that $\wri(\ell,0) \xrightarrow{\tw} \wri(\ell,1)$.
	If $\ell$ is a literal evaluating to $0$ under $v$, we can obtain a similar proof.
	This shows the above equivalence.
	
	Now we show that for each clause, there is at least one literal that evaluates to $1$ under $v$.
	Assume the contrary, then there is a clause $C = \ell_1 \vee \ell_2 \vee \ell_3$ such that $v(\ell_i) = 0$ for $i = 1,2,3$.
	By the equivalence above, we obtain $\wri(\ell_i,1) \xrightarrow{\tw} \wri(\ell_i,0)$ for each $i = 1,2,3$.
	From this, we can obtain the following cycle in $\graph{sc}$:
	\begin{align*}
	\rd(\ell_1,0) \xrightarrow{\po} \rd(\ell_2,1)
	\xrightarrow{\cf} \wri(\ell_2,0)
	\xrightarrow{\rf} \rd(\ell_2,0) \\
	\xrightarrow{\po} \rd(\ell_3,1)
	\xrightarrow{\cf} \wri(\ell_3,0)
	\xrightarrow{\rf} \rd(\ell_3,0) \\
	\xrightarrow{\po} \rd(\ell_1,1)
	\xrightarrow{\cf} \wri(\ell_1,0)
	\xrightarrow{\rf} \rd(\ell_1,0).
	\end{align*}
	Hence, the clauses are satisfied under $v$.
	This completes the proof.
	\qedhere
\end{proof}
\begin{figure}[h]
	\begin{center}
	\begin{tikzpicture}
		\node  (Tell0) {$T_0(\ell):$};
		\node [below = -0.1cm of Tell0] (Tell0rdx0a) {$\rd(x,0)$};
		\node [below = -0.1cm of Tell0rdx0a] (Tell0wrell) {$\wri(\ell,c)$};
		
		\node [right = 0.35cm of Tell0] (Tell1) {$T_1(\ell):$};
		\node [below = -0.1cm of Tell1] (Tell1rdx1a) {$\rd(x,1)$};
		\node [below = -0.1cm of Tell1rdx1a] (Tell1wrell) {$\wri(\ell,d)$};
		
		\node [right = 0.35cm of Tell1](Tx0) {$T_0'(\ell):$};
		\node [below = -0.1cm of Tx0] (wrx0) {$\rd(\ell,c)$};
		\node [below = -0.1cm of wrx0] (wrx01) {$\rd(x,0)$};
		
		\node [right = 0.35cm of Tx0] (Tx1) {$T_1'(\ell):$};
		\node [below = -0.1cm of Tx1] (wrx1) {$\rd(\ell,d)$};
		\node [below = -0.1cm of wrx1] (wrx11) {$\rd(x,1)$};
	\end{tikzpicture}
\end{center}
	\captionof{figure}{Parts of the history $h'_\varphi$ for  a literal $\ell \in L$.
		Values of $c$ and $d$ depend on $\ell$.
		If $\ell = x$, then $c = 0, d = 1$. Otherwise, $c = 1, d = 0$.}
	\label{Figure:LowerBoundPSOTSO}
\end{figure}
\begin{proof}[Proof of Lemma \ref{Thorem:TSOPSOLower}]
	Recall that for the memory models $\TSO$ and $\PSO$, the  preserved program orders are $\potso = \po \! \setminus \WR \times \! \RD$ and $\popso = \po \! \setminus \! ( \WR \! \times \! \RD \cup \WR \! \times \! \WR )$, respectively.
	While there are no $\po$-relations of the form $ \WR \! \times \! \WR$ in our lower bound construction $h_\varphi$ for the case of $\SC$, it has $\po$-relations of the form $ \WR \! \times \! \RD$.
	Recall that the threads $T_0(\ell)$ and $T_1(\ell)$ were guarded by reads. 
	The program order edges connecting the write event in the threads with the latter read events will vanish under $\TSO$ and $\PSO$.
	Hence, for a total order $\tw$ on the write events of $h_\varphi$, a cycle in $\graph{sc}$ does not necessarily imply a cycle on $\graph{tso}$ or $\graph{pso}$.
	
	We overcome this issue by replacing the latter guard in $T_i(\ell)$ by a separate thread $T_i'(\ell)$.
	The thread reads the value written to $\ell$ in $T_i'(\ell)$, followed by the guarding read.
	The construction is shown in Figure \ref{Figure:LowerBoundPSOTSO}.
	We denote the obtained history by $h'_\varphi$.
	The advantage of $h'_\varphi$ is, that for any total order $\tw$, the graphs $\graph{sc}$, $\graph{tso}$, and $\graph{pso}$ are all equal.
	This is due to the construction.
	There are no $\po$-relations of the form $\WR \! \times \! \RD $ or $\WR \! \times \! \WR$ which could be relaxed by $\TSO$ or $\PSO$.
	Intuitively, we enforce sequential behavior with the new threads.
	
	It is left to prove that $\varphi$ is satisfiable if and only if $h'_\varphi$ is $\SC$-consistent (and thus $\TSO / \PSO$-consistent).
	To this end, note that $h_\varphi$ is $\SC$-consistent if and only if $h'_\varphi$ is.
	Indeed, the program order edge $\wri(\ell,c) \xrightarrow{\po} \rd(x,0)$ in $h_\varphi$ is replaced by $\wri(\ell,c) \xrightarrow{\rf} \rd(\ell,c) \xrightarrow{\po} \rd(x,0)$ in $h'_\varphi$.
	A similar replacement is done for $\wri(\ell,d) \xrightarrow{\po} \rd(x,1)$.
	Hence, there is a path $\wri(\ell,c) \rightarrow^* \rd(x,0)$ in $h_\varphi$ if and only if there is a path $\wri(\ell,c) \rightarrow^* \rd(x,0)$ in $h'_\varphi$.
	Due to this, acyclicity (and hence $\SC$-consistency) is preserved across the histories.
	\qedhere
\end{proof}

\end{document}